\documentclass[11pt]{article}

\usepackage{fullpage}
\usepackage{amsmath}
\usepackage{amsthm}
\usepackage{subfigure}
\usepackage{enumerate}

\usepackage{pst-all}
\usepackage{natbib}

\usepackage{latexsym}
\usepackage{amsmath}
\usepackage{amstext}
\usepackage{amsfonts}
\usepackage{amsopn}

\usepackage{ifthen}

\newtheorem{theorem}{Theorem}[section]
\newtheorem{lemma}[theorem]{Lemma}
\newtheorem{corollary}[theorem]{Corollary}
\newtheorem{proposition}[theorem]{Proposition}
\newtheorem{definition}[theorem]{Definition}

\newenvironment{numberedtheorem}[1]{%
\renewcommand{\thetheorem}{#1}%
\begin{theorem}}{\end{theorem}\addtocounter{theorem}{-1}}

\newenvironment{numberedlemma}[1]{%
\renewcommand{\thetheorem}{#1}%
\begin{lemma}}{\end{lemma}\addtocounter{theorem}{-1}}

\newcommand{\reals}{\mathbb R}

\newcommand{\prob}[2][]{\text{\bf Pr}\ifthenelse{\not\equal{}{#1}}{_{#1}}{}\!\left[#2\right]}
\newcommand{\expect}[2][]{\text{\bf E}\ifthenelse{\not\equal{}{#1}}{_{#1}}{}\!\left[#2\right]}
\newcommand{\given}{\,\mid\,}

\newcommand{\setsize}[1]{\left| #1 \right|}

\DeclareMathOperator{\argmax}{argmax}

\newcommand{\super}[1]{^{(#1)}}

\newcommand{\Xcomment}[1]{}

\newcommand{\mech}{{\cal M}}
\DeclareMathOperator{\RSOL}{RSOL}

\newcommand{\bid}{b}

\newcommand{\bidi}[1][i]{\bid_{#1}}

\newcommand{\val}{v}
\newcommand{\vals}{{\mathbf \val}}
\newcommand{\valsmi}{{\mathbf \val}_{-i}}
\newcommand{\vali}[1][i]{\val_{#1}}
\newcommand{\valith}[1][i]{\val_{(#1)}}

\newcommand{\util}{u}

\newcommand{\utili}[1][i]{\util_{#1}}

\newcommand{\virt}{\varphi}

\newcommand{\dist}{F}

\newcommand{\disti}[1][i]{{\dist_{#1}}}

\newcommand{\dens}{f}

\newcommand{\densi}[1][i]{{\dens_{#1}}}

\newcommand{\price}{p}
\newcommand{\prices}{{\mathbf \price}}
\newcommand{\pricei}[1][i]{\price_{#1}}

\newcommand{\alloc}{x}
\newcommand{\allocs}{{\mathbf \alloc}}

\newcommand{\alloci}[1][i]{\alloc_{#1}}

\newcommand{\feasibles}{{\cal X}}

\newcommand{\expval}{\mu}

\newcommand{\marg}{\vartheta}

\newcommand{\ironmarg}{\bar{\marg}}

\newcommand{\highprice}{p}
\newcommand{\lowprice}{q}
\newcommand{\lowprob}{\rho}
\newcommand{\oneprice}{r}

\newcommand{\valweight}{\gamma_\val}
\newcommand{\priceweight}{\gamma_\price}

\newcommand{\ratio}{\alpha}
\newcommand{\bias}{\chi}

\newcommand{\bm}{{\cal G}}
\newcommand{\of}{{\cal F}}
\newcommand{\bms}{\bm \super 2}
\newcommand{\ofs}{\of \super 2}
\DeclareMathOperator{\Mye}{Opt}
\DeclareMathOperator{\Lottery}{Lot}

\newcommand{\A}{\mathcal{A}}
\newcommand{\C}{\mathcal{C}}

\newcommand{\pwcfunc}{\marg}
\newcommand{\interval}{a}
\newcommand{\exval}{\mu}
\newcommand{\lbparam}{\beta}
\newcommand{\lbind}{\kappa}

\begin{document}

\title{Optimal Platform Design\thanks{There is some overlap between
    this paper and the paper ``Optimal Mechanism Design and Money
    Burning,'' which appeared in the STOC 2008 conference.  However,
    the focus of this paper is different, with some of our earlier
    results omitted and several new results included.}} 

\author{Jason D. Hartline\thanks{Electrical Engineering and Computer
Science, Northwestern University, Evanston, IL 60208.
Email: {\tt hartline@eecs.northwestern.edu}. This work was done
while author was at Microsoft Research, Silicon Valley.} \and
Tim Roughgarden\thanks{Department of Computer Science,  Stanford
University, Stanford, CA 94305.
Email: {\tt tim@cs.stanford.edu}.
  Supported in part by NSF CAREER Award CCF-0448664, an ONR Young
Investigator Award, and an Alfred P. Sloan Fellowship.}}

\date{}



\maketitle

\begin{abstract}
An auction house cannot generally provide the optimal auction
technology to every client.  Instead it provides one or several
auction technologies, and clients select the most appropriate one.
For example, eBay provides ascending auctions and ``buy-it-now''
pricing.  For each client the offered technology may not be optimal,
but it would be too costly for clients to create their own.  We call
these mechanisms, which emphasize generality rather than optimality,
platform mechanisms.  A platform mechanism will be adopted by a client
if its performance exceeds that of the client's outside option, e.g.,
hiring (at a cost) a consultant to design the optimal
mechanism.  We ask two related questions.  First, for what costs of
the outside option will the platform be universally adopted?  Second,
what is the structure of good platform mechanisms?  We answer these
questions using a novel prior-free analysis framework in which we seek
mechanisms that are approximately optimal for every prior.
\end{abstract}

\section{Introduction}
\label{sec:intro}

Auction houses, like Sotheby's, Christie's, and eBay, exemplify the
commodification of economic mechanisms, like auctions, and warrant an
accompanying theory of design.  The field of {\em mechanism
  design} suggests how special-purpose mechanisms might be optimally
designed; however, in commodity industries there is a trade-off
between special-purpose and general-purpose products.  While for any
particular setting an optimal special-purpose product is better, a
general-purpose product may be favored, for instance, because of its
cheaper cost or greater versatility.  We develop a theory for
the design of general-purpose mechanisms, henceforth, {\em platform
  design}.

Consider the following simple model for platform design.  The platform
{\em provider} offers a {\em platform} mechanism to potential
customers ({\em principals}), who each wish to employ the mechanism in
their particular {\em setting}.  For example, the provider is eBay,
the platform is the eBay auction, the principals are sellers, and the
settings are the distinct markets of the sellers, which comprise of a
set of buyers ({\em agents}) with preferences drawn according to a
distribution.  Each principal has the option to not adopt the platform
and instead to employ a consultant to design the optimal auction for
his specific setting.  We assume that this outside option comes at a
greater cost than the platform, and thus the platform provider has a
{\em competitive advantage}.

We impose two restrictions to focus on the differences between the
special-purpose optimal mechanism design and the general-purpose
optimal platform design.  First, we restrict the platform to be a
single, unparameterized mechanism (unlike eBay where sellers can set
their own reserve prices).\footnote{In a separate study, we consider
  the technically orthogonal topic of reserve-price based
  platforms~\citep{HR-09}.}  Second, we require that the platform is
universally adopted.  Without this assumption, we would need to model
in detail the relative value of adoption in each setting, and this
would likely give less general results.  We ask: {\em What must the
  competitive advantage of the platform be to guarantee universal
  adoption by all principals?  What is the platform designer's
  mechanism that guarantees universal adoption?}

There are two important points of contact between this theory of platform
design and the existing literature.  First, the problem of optimal
platform design provides a formal setting in which to explore the
\citet{wil-87} doctrine, which critiques mechanisms that are overly dependent on
the details of the setting but does not quantify the cost of
this dependence.  
%
%
A universally adopted platform, by definition, performs well in all
settings and hence is not dependent on the details of setting.
Second, the optimal platform design problem is closely
related to {\em prior-free optimal mechanism design}.  Indeed, our
study of platform design formally connects the prior-free and Bayesian
theories of optimal mechanism design.  We make a rigorous comparison
between the two settings and quantify the Bayesian designer's relative
advantage over the prior-free designer.

\paragraph{Platform Design.}
%
In classical Bayesian optimal mechanism design, a principal
designs a mechanism for a set of self-interested agents that have
private preferences over the outcomes of the mechanism.  These private
preferences are drawn from a known probability distribution.
The optimal mechanism is the one that maximizes the expected value of
the principal's objective function when the agents' strategies are in
Bayes-Nash equilibrium.


For a given distribution and objective function, 
the {\em approximation factor} of a candidate mechanism is the ratio 
between the expected performance of an optimal mechanism and that of
the candidate mechanism.
A good mechanism is one with a small approximation factor (close
to~1); a bad one has a large approximation factor.  

We assume that the cost of designing
the optimal mechanism is higher than the cost of adopting the
platform.  For this reason, a principal might choose to adopt the
sub-optimal platform mechanism.  We assume this {\em competitive
  advantage} of the platform is multiplicative.  This assumption is
consistent with commission structures in marketing and, from a
technical point of view, frees the model from artifacts of scale.  The
platform's competitive advantage gives an upper bound on the
approximation factor that the platform mechanism needs to induce a
principal to adopt the platform instead of hiring a consultant to design
the optimal mechanism.
Each principal's decision to adopt is based on the platform mechanism's
performance in the principal's setting.  Therefore, universal
adoption demands that the platform mechanism's approximation factor
on every distribution is at most its competitive advantage.  Of
particular interest is the minimum competitive advantage for which
there is a platform that is universally adopted, and also the
platform that attains this minimum approximation factor.  This {\em
  optimal platform} is the mechanism that minimizes (over mechanisms) the
maximum (over distributions) approximation factor.  Optimal platform
design is therefore inherently a min-max design criterion.

The basic formal question of platform design is: {\em What is the
  minimum competitive advantage $\beta$ and optimal platform mechanism
  $\mech$ such that for all distributions $\dist$ the expected
  performance of $\mech$ when values are drawn i.i.d.~from $\dist$ is
  at least $\frac{1}{\beta}$ times the expected performance of the
  optimal mechanism for $\dist$?}

%

Directly answering the platform design questions above is
difficult as it requires simultaneous consideration of all
distributions.  This difficulty motivates a more stringent version of
the basic question which has the following economic interpretation.
Suppose that instead of requiring the principal to choose ex ante
between the optimal mechanism and the platform, we allow him to choose
ex post?  Clearly, this makes the platform designer's task even more
challenging, in that the minimum achievable $\beta$ is only higher.

The formal question of platform design now becomes: {\em What is the
  minimum competitive advantage $\beta$ and optimal platform mechanism
  $\mech$ such that for all valuation profiles $\vals =
  (\vali[1],\ldots,\vali[n])$ the performance of $\mech$ on $\vals$ is
  at least $\frac{1}{\beta}$ times the supremum over symmetric\footnote{Our study focuses solely on settings where the agents are a
  priori indistinguishable.  This focus motivates our restriction to
  i.i.d.\ distributions and symmetric optimal mechanisms.
  Distinguishable agents are considered by \citet{BBHM-08} and
  \citet{B+-13}.}
 Bayesian optimal
  mechanisms' performance on $\vals$?}

This question motivates the definition of a {\em performance
  benchmark} that is defined point-wise on valuation profiles,
specifically as the supremum over optimal symmetric mechanisms'
performance on the given valuation profile.  Notice that this
benchmark is prior-free.  The analysis of a platform mechanism is then
a comparison of the performance of a prior-free platform mechanism and
a prior-free performance benchmark.


\paragraph{Results.}

Our contributions are two-fold.  First, we propose a conceptual
framework for the design and analysis of general-purpose platforms.
Second, we instantiate this framework to derive novel platform
mechanisms for specific problems and, in some cases, prove their
optimality. 

In more detail, we consider the problem of optimal platform design in
general symmetric settings of multi-unit unit-demand allocation
problems and for general linear (in agents' payments and values)
objectives of the principal.  For much of the paper, we focus on the
canonical objective of {\em residual surplus}, which is the difference
between the winning agents' values and payments.  Residual surplus is
interesting in its own right \citep[e.g.,][]{MM-92,con-12,CK-12} and
is, in a sense, technically more general than the objectives of
surplus and profit.\footnote{For surplus maximization, the Vickrey
  auction is optimal for every distribution.  For profit maximization,
  reserve-price-based auctions are optimal for standard
  distributions assumptions~\citep{mye-81}.  For
  residual surplus, reserve-price-based auctions are not optimal even
  for standard distributions.}  Intuitively, maximizing the residual
surplus involves compromising between the competing goals of
identifying high-valuation agents and of minimizing payments.  For
example, with a single item, the Vickrey auction performs well when
there is only one high-valuation agent, while giving the item away for
free is good when all agents have comparable valuations.
%

Our approach comprises four steps.
\begin{enumerate}

\item We characterize Bayesian optimal mechanisms for multi-unit
  unit-demand allocation problems and general linear objectives by a
  straightforward generalization of the literature on optimal
  mechanism design.

\item We characterize the prior-free performance benchmark, i.e., the supremum
  over optimal symmetric mechanisms' performance on a given
  valuation profile, as an ex post optimal two-level lottery.

\item \label{result:upper-bounds} We give a general platform design
  and a finite upper bound on the 
  competitive advantage necessary for universal adoption.  

\item We give a lower bound on the competitive advantage for which
  there exists a platform that achieves universal adoption.

\end{enumerate}

Importantly, the platform mechanisms that we identify as being
universally adopted with finite competitive advantage are not standard
mechanisms from the literature on Bayesian optimal mechanisms.
Indeed, we prove that no standard mechanism is universally adopted
with any finite competitive advantage.  Instead, general purpose
mechanisms for platforms require novel features, which we identify in
Step~\ref{result:upper-bounds}.


\paragraph{Example.}  Our main results are 
interesting to interpret in the special case of allocating a single
item to one of two agents to maximize the residual surplus.  Denote
the high agent value by $\valith[1]$ and the low agent value by
$\valith[2]$.  We characterize the performance benchmark as $\max
(\tfrac{\valith[1] + \valith[2]}{2}, \valith[1] -
\tfrac{\valith[2]}{2})$.  As the supremum of Bayesian optimal
mechanisms, the first term in this benchmark arises from a lottery and
the second term from the two-level lottery that serves a random agent
with value strictly above price $\valith[2]$ if one exists and
otherwise serves a random agent (at price zero).  The optimal platform
mechanism randomizes between a lottery and weighted Vickrey auctions.
Precisely, it sets $w_1 = 1$, draws $w_2$ uniformly from
$\{0,1/2,2,\infty\}$, and serves the agent $i \in \{1,2\}$ that
maximizes $w_i \vali$.  It is universally adopted with
competitive advantage $\frac43$ and no other mechanism is better.

While all possible prior distributions are considered when deriving
the performance benchmark above, the actual benchmark for a particular
valuation profile is given by a simple formula with no distributional
dependence.  Consequently, our analysis that shows the $\frac43$
competitive advantage is a simple comparison between a (prior-free)
platform mechanism and a (prior-free) performance benchmark in the
worst case over valuation profiles.

\paragraph{Related Work.} 

Our description of Bayesian optimal mechanisms for general linear
objectives follows from the work on optimal mechanism design
(see \citealp{mye-81}, and \citealp{RS-81}).  Within this theory, the
residual surplus objective coincides with that of the grand coalition
in a {\em weak cartel}, where agents wish to maximize the cartel's
total utility without side payments amongst themselves, so that
payments to the auctioneer are effectively ``burnt''.  Our
characterizations are thus related to those in the literature on
collusion in multi-unit auctions, e.g., by \citet{MM-92} and
\citet{con-12}.  Recently, \citet{CK-12} also specifically studied
Bayesian optimal auctions for residual surplus.

There is a growing literature on ``redistribution mechanisms'' where,
similar to the objective of residual surplus, payments are bad, e.g.,
see \citet{mou-06} and \citet{GC-09}.  These mechanisms transfer some
of the winners' payments back to the losers so that the residual
payment left over is as small as possible.  The mechanisms considered
are prior-free.

Finally, as already mentioned, there is a large related literature on
prior-free optimal mechanism design.  \citet{GHW-01}, \citet{seg-03},
\citet{BV-03}, and \citet{BBHM-08} consider asymptotic approximation
of the Bayesian optimal mechanisms by a single (prior-free) mechanism.
This is quite different from our question of platform design as it
says nothing about whether or not a principal in a small or
moderate-sized market would adopt the platform.  The line of research
initiated by \citet{GHKSW-06} on prior-free profit maximization can be
reinterpreted in the context of platform design;
Section~\ref{sec:profit} describes this connection in detail.
%


\section{Warm-up: Monopoly Pricing}
\label{sec:monop-pricing}

Consider the following monopoly pricing problem.  A monopolist seller
(principal) of a single item faces a single buyer (agent).  The seller
has no value for the item and wishes to maximize his revenue, i.e.,
the payment of the buyer.  The buyer's value for the item is $\val \in
[1,h]$ and she wishes to maximize her utility which is her value less
her payment.  The seller may post a price $\price$ and the buyer may
take it or leave it.  The buyer will clearly take any price $\price
\leq \val$.

The seller's optimal mechanism, when the buyer's value comes from the
distribution $\dist$ (where $\dist(z) = \prob{\val \le z}$), is to
post the price $\price$ that maximizes $\price (1-\dist(\price))$,
a.k.a., the {\em monopoly price}.  The performance benchmark
$\bm(\val)$, i.e., the revenue of the best of the Bayesian optimal
mechanism when the buyer's value is $\val$, is then $\bm(\val) =
\val$.  The platform designer must give a single mechanism with
revenue that approximates $\val$ for every value $\val$ in the support
$[1,h]$.  The optimal platform and its competitive advantage for
universal adoption are given by the theorem below.

\begin{theorem}
The optimal platform mechanism offers a price drawn from distribution
$P$ with cumulative distribution function $P(z) = (1 + \ln z)/(1 + \ln
h)$ on $[1,h]$, and a point mass of $1/(1+\ln h)$ at $1$, and is
universally adopted with competitive advantage $1 + \ln h$.
\end{theorem}

\begin{proof}
An easy calculation verifies that, for every $\val \in [1,h]$, the
expected revenue from such a random price from $P$ is $\val/(1 + \ln
h)$.  Thus, the competitive advantage for universal adoption is $1+\ln
h$ as claimed.

To show that this is the optimal platform, we can similarly find a
distribution $\dist$ over values $\val$ such that the expected revenue
of every platform mechanism is 1.  The {\em equal revenue}
distribution has distribution function $\dist(z) = 1-1/z$, a point
mass of $1/h$ at $h$, and any price $\price$ is accepted by the agent
with probability $1/\price$ for an expected revenue of 1.  The
expected value of the benchmark for the equal-revenue distribution can
be calculated as $\expect{\bm(\val)} = \expect{\val} = 1+\ln h$.
Thus, the ratio of these expectations is $1+\ln h$, and for any
platform mechanism there must be some $\val \in [1,h]$ that achieves
the ratio.  We conclude that no platform is universally adopted with
competitive advantage less than $1+\ln h$.
\end{proof}

This analysis can be viewed as a zero-sum game between the platform
designer and Nature where the solution is a mixed strategy on the part
of both players, every action in the game achieves equal payoff, and
the value of the game is the optimal competitive advantage.

To conclude, we considered a simple monopoly pricing setting and
derived for it the optimal platform.  While a logarithmic competitive
advantage may seem impractical, except when the maximum variation~$h$
in values is small, the ideas from this design and analysis play an
important role in the developments of this paper.  The platform
mechanisms we derive subsequently, however, will be universally
adopted with a competitive advantage that is an absolute constant,
independent of the number of agents, the number of units, and the
range of agent values.

\section{Review of Bayesian Optimal Mechanism Design}
\label{sec:bayesian}

In this section we review Bayesian optimal mechanism design for
single-dimensional agents, i.e.,  with utility given by the value for
receiving a good or service less the required payment, and develop the
notation employed in the remainder of the paper.  Characterizing
Bayesian optimal mechanisms is the first step in our approach to
platform design.

We consider mechanisms for allocating $k$ units of an indivisible item
to $n$ unit-demand agents.  The outcome of such a mechanism is an {\em
  allocation vector}, $\allocs = (\alloci[1],\ldots,\alloci[n])$,
where $\alloci$ is~1 if agent~$i$ receives a unit and~0 otherwise, and
a non-negative {\em payment vector}, $\prices = (\pricei[1],\ldots,\pricei[n])$.
The allocation vector $\allocs$ is required to be feasible, i.e.,
$\sum_i \alloci \leq k$, and we denote this set of feasible allocation vectors
by $\feasibles$.

We assume that each agent~$i$ is risk-neutral, has a privately
known valuation $\vali$ for receiving a unit, and aims to maximize
her (quasi-linear) utility, defined as $\utili = \vali \alloci -
\pricei$.  Each agent's value is drawn independently and identically
from a continuous distribution $\dist$, where $\dist(z)$ and $\dens(z)$
denote the cumulative distribution and density functions,
respectively.  We denote the {\em valuation profile} by $\vals =
(\vali[1],\ldots,\vali[n])$.

We consider general symmetric, linear objectives of the mechanism
designer.  For valuation
coefficient $\valweight$ and payment coefficient $\priceweight$, the objective for maximization is:
\begin{align}
\label{eq:objective}
\sum\nolimits_{i=1}^n \valweight \vali \alloci +
\priceweight \pricei.
\end{align}  
We single out three such objectives:
{\em surplus} with $\valweight = 1$ and
$\priceweight = 0$, {\em profit} with $\valweight = 0$ and
$\priceweight = 1$, and {\em residual surplus} with $\valweight = 1$
and $\priceweight = -1$.  We will not discuss surplus maximization in
this paper as the optimal mechanism for this objective is 
simply the prior-free $k$-unit
Vickrey auction; therefore, we assume that $\priceweight\neq 0$.

We assume that agents play in Bayes-Nash equilibrium and moreover if
truthtelling is a Bayes-Nash equilibrium then agents truthtell.  When
searching for Bayesian optimal mechanisms, the revelation
principle~\citep{mye-81} allows us to restrict attention to Bayesian
incentive compatible mechanisms, i.e., ones with a truthtelling
Bayes-Nash equilibrium.

\paragraph{Characterization of incentive compatibility.}
The {\em allocation rule}, $\allocs(\vals)$, is the mapping (in
equilibrium) from agent valuations to the outcome of the mechanism.
Similarly the {\em payment rule}, $\prices(\vals)$, is the mapping
from valuations to payments.  Given an allocation rule
$\allocs(\vals)$, let $\alloci(\vali)$ be the interim probability with
which agent $i$ is allocated when her valuation is $\vali$ (over the
probability distribution on the other agents' valuations):
$\alloci(\vali) = \expect[\valsmi]{\alloci(\vali,\valsmi)}.$ Similarly
define $\pricei(\vali)$.  We require interim individual rationality,
i.e., that non-participation in the mechanism is an allowable agent
strategy.  The following lemma provides the standard characterization
of allocation rules that are implementable by Bayesian incentive
compatible mechanisms and the accompanying payment rule (which is
unique up to additive shifts, and usually fixed by setting $\pricei(0)
= 0$).
\begin{lemma}\citep{mye-81} 
\label{l:bic}\label{l:ic} 
Every Bayesian incentive compatible mechanism satisfies, for all $i$ and $\vali \geq \vali'$:
\begin{enumerate}[(a)]
\item Allocation monotonicity: 
$\alloci(\vali) \geq \alloci(\vali').$
\item Payment identity: 
$\pricei(\vali) = \vali\alloci(\vali) - \int_0^{\vali}
  \alloci(z)\, dz + \pricei(0)$.
\end{enumerate}
\end{lemma}

\paragraph{Virtual valuations.}
\citet{mye-81} defined {\em virtual valuations} and showed that the
virtual surplus of an agent is equal to her expected payment.  For
$\val\sim\dist$, this {\em virtual valuation for payment} is:
\begin{align}
\label{eq:virtual-value-for-payment}
\virt(\vali) &= \vali - \tfrac{1-\dist(\vali)}{\dens(\vali)}.
\end{align}

\begin{lemma}\citep{mye-81}
\label{l:vv} 
In a Bayesian incentive-compatible mechanism with allocation rule
$\allocs(\cdot)$, the expected payment of an agent equals her expected
virtual surplus:
$\expect[\vals] {\pricei(\vals)} = \expect[\vals] {\virt(\vali)\,\alloci(\vals)}.$
\end{lemma}
The notion of virtual valuations applies generally to linear
objectives.  By substituting virtual values for payments into the
objective~\eqref{eq:objective} we arrive at a formula for {\em general
  virtual values}: 
$\marg(\vali) = (\valweight+\priceweight)\vali -
\priceweight \tfrac{1 - \disti(\vali)}{\densi(\vali)}.$
For the
objective of residual surplus, i.e., the sum of the agent utilities,
{\em virtual values for utility} are given by: 
\begin{align}
\label{eq:virtual-value-for-utility}
\marg(\vali) &=
\tfrac{1-\disti(\vali)}{\densi(\vali)}.
\end{align}

The revenue-optimal mechanism for a given distribution is the one that
maximizes the virtual surplus for payment subject to feasibility and
monotonicity of the allocation rule.
Analogously, optimal mechanisms for general linear objectives are
precisely those that maximize the expected (general) virtual surplus
subject to feasibility and monotonicity of the allocation rule.
%
%
Unfortunately, choosing $\allocs$ to maximize $\sum_i \marg(\vali)
\alloci$ for each valuation profile $\vals$ does not generally result
in a monotone allocation rule.  When $\marg(\cdot)$ is not monotone
increasing, an increase in an agent's value may decrease her virtual
value and cause her to be allocated less frequently.  Notice that
under the standard ``monotone hazard rate'' assumption the virtual
value function for utility $\marg(\val) =
\frac{1-\disti(\val)}{\densi(\val)}$ is monotone {\em in the wrong
  direction}.

\paragraph{Ironing.}
We next generalize the ``ironing'' procedure of \citet{mye-81} that
transforms a possibly non-monotone virtual valuation function into an
{\em ironed virtual valuation} function that is monotone; optimizing
ironed virtual surplus results in a monotone allocation rule.
Furthermore, the ironing procedure preserves the target objective, so
that an optimal allocation rule for the ironed virtual valuations is
equal to the optimal monotone allocation rule for the original virtual
valuations.

Given a distribution function $\dist(\cdot)$ with virtual valuation
function $\marg(\cdot)$, the {\em ironed virtual valuation function},
$\ironmarg(\cdot)$, is constructed as follows:
\begin{enumerate}
\item For $q \in [0,1]$, define $h(q) = \marg(\dist^{-1}(q))$.
\item Define $H(q) = \int_0^q h(r) dr$.
\item Define $G$ as the convex hull of~$H$ --- the largest
convex function bounded above by~$H$ for all $q \in [0,1]$.
\item Define $g(q)$ as the derivative of $G(q)$, where defined, 
  extended to all of $[0,1]$ by right-continuity. 
\item Finally, define $\ironmarg(z) = g(\dist(z))$.
\end{enumerate}
Convexity of $G$ implies that Step~4 of the ironing procedure is well
defined and that $g$, and hence~$\ironmarg$, is a monotone
non-decreasing function.

From the main theorem of \citet{mye-81}, maximizing the expectation of
a general linear objective subject to incentive compatibility is
equivalent to maximizing the expected ironed virtual surplus.
Different tie-breaking rules, however, can yield different optimal mechanisms.
In our symmetric settings, with i.i.d.\@ agents and the symmetric
feasibility constraint $\feasibles$ of $k$-unit auctions, it is
natural to consider symmetric optimal mechanisms.

\begin{theorem} 
\label{t:main} 
For every general linear objective and distribution $\dist$, the
$k$-unit auction that allocates the units to the agents with the
highest non-negative ironed virtual values, breaking ties randomly and
discarding all leftover units, maximizes the expected value of the objective.
\end{theorem}

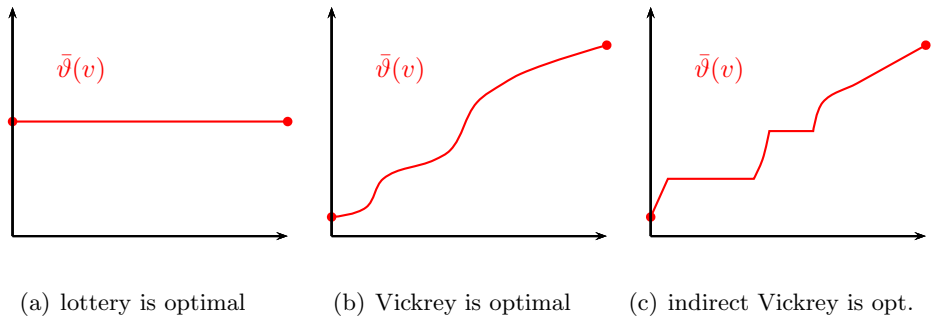
\begin{figure}[t]
\begin{center}
\psset{yunit=1in,xunit=.9in}
\subfigure[lottery is optimal]{\small
\begin{pspicture}(-.2,-.2)(1.6,1.3)

\rput[b](.4,.8){\red$\ironmarg(\val)$}

\psline[linecolor=red]{*-*}(0,.6)(1.6,.6)

\psaxes[labels=none,ticks=none]{->}(1.6,1.2)

\end{pspicture}}
%
%
\subfigure[Vickrey is optimal]
{\small
\begin{pspicture}(-.2,-.2)(1.6,1.3)

\rput[b](.4,.8){\red$\ironmarg(\val)$}

\pscurve[linecolor=red]{*-*}(0,.1)(.2,.15)(.3,.3)(.6,.4)(.69,.45)(.845,.70)(1,.8)(1.1,.85)(1.6,1)

\psaxes[labels=none,ticks=none]{->}(1.6,1.2)

\end{pspicture}}
%
%
\subfigure[indirect Vickrey is opt.]
{\small
\begin{pspicture}(-.2,-.2)(1.6,1.3)

\rput[b](.4,.8){\red$\ironmarg(\val)$}

\pscurve[linecolor=red]{*-*}(0,.1)(.1,.3)(.1,.3)(.6,.3)(.6,.3)(.65,.4)(.69,.55)(.69,.55)(.945,.55)(.945,.55)(1,.7)(1.2,.8)(1.3,.85)(1.6,1)

\psaxes[labels=none,ticks=none]{->}(1.6,1.2)

\end{pspicture}}
\end{center}
\caption{Ironed virtual value functions in the three distributional
  cases.  For the objective of residual surplus the cases correspond
  to (a) MHR distributions,
(b) anti-MHR distributions,
and 
(c) non-MHR distributions.
} 
\label{fig:ironed-residual-cases}
\end{figure}

\paragraph{Interpretation for residual surplus maximization.}
Consider the residual surplus objective, where $\marg(\val) =
\frac{1-\dist(\val)}{\dens(\val)}$, and the following three types of
distributions (Figure~\ref{fig:ironed-residual-cases}).  Monotone
hazard rate (MHR) distributions; e.g., uniform, normal, and
exponential; have monotone non-increasing $\marg(\val)$.  In this
case, ironing $\marg(\cdot)$ to be non-decreasing results in
$\ironmarg(\cdot) = \expect{\val}$, a constant function.  The optimal
(symmetric) mechanism is therefore a lottery that awards the $k$ units
to $k$ agents uniformly at random.  For distributions with a hazard
rate monotone in the opposite direction, henceforth {\em anti-MHR}
distributions, $\marg(\cdot)$ is non-negative and monotone
non-decreasing.  Power-law distributions, such as $\dist(z) = 1 -
1/z^c$ with $c > 0$ on $[1,\infty)$, are canonical examples.  In this
  case, the optimal mechanism awards the $k$ units to the $k$ highest
  valued agents, i.e., it is the $k$-Vickrey auction. Thus, as also
  observed by \citet{MM-92}, \citet{CK-06}, and \citet{con-07}, the
  optimal mechanism depends on whether or not the distribution is
  heavy-tailed.

The final case occurs when the distribution is neither MHR nor
anti-MHR, henceforth {\em non-MHR}.  Here, the ironed virtual
valuation function $\ironmarg(\cdot)$ is constant on some intervals
and monotone increasing on other intervals.  The optimal mechanism can
be described, for instance, as an indirect Vickrey auction where
agents are not allowed to bid on intervals where the ironed virtual
value is constant.  For example, consider the two-point distribution
with probability mass $\tfrac{1}{2}$ on~1 and~$\tfrac{1}{2}$ on~$h >
1$.  Provided~$h$ is sufficiently large, the
residual-surplus-maximizing mechanism allocates to a random high-value
agent or, if there are no high-value agents, to a random (low-value)
agent.  This final case is the most general, in that it subsumes both
the MHR and anti-MHR cases. Our general theory of platform design
necessitates understanding this non-MHR case in detail.

\section{The Performance Benchmark}
\label{sec:bm}\label{sec:benchmark}

In this section we leverage the characterization of Bayesian optimal
mechanisms from the preceding section to identify and characterize a
simple prior-free performance benchmark.  This constitutes the
second step of our approach to platform design.

The performance benchmark is derived as follows.  As discussed in
Section~\ref{sec:bayesian}, Bayesian optimal mechanisms are ironed
virtual surplus optimizers.  For $k$-unit environments, these
mechanisms simply select the $k$ agents with the highest non-negative
ironed virtual values.  Among these optimal mechanisms, the symmetric
one breaks ties randomly.  Denote the symmetric optimal mechanism for
distribution $\dist$ by $\Mye_\dist$.  Denote by $\Mye_\dist(\vals)$
the expected performance (over the choice of random allocation)
obtained by the mechanism $\Mye_\dist$ on the valuation profile
$\vals$.

\begin{definition} 
\label{d:benchmark}
The {\em performance benchmark} is the supremum of Bayesian optimal
mechanisms, 
$\bm(\vals) = \sup\nolimits_\dist \Mye_\dist(\vals).$
\end{definition}

For one interpretation of the definition of $\bm$, observe that
\begin{equation}\label{eq:bm2}
\expect[\vals]{\bm(\vals)} \geq \expect[\vals]{\Mye_\dist(\vals)}
\end{equation}
for valuation profiles drawn i.i.d.~from an arbitrary distribution
$\dist$.  
Thus, the approximation of the
performance benchmark $\bm$ implies the simultaneous
approximation of all 
symmetric Bayesian optimal mechanisms.

We now give a simple characterization of the performance benchmark for
general linear objectives by considering ex post outcomes of symmetric
Bayesian optimal mechanisms.  When $k$ units are available, a
symmetric Bayesian optimal mechanism serves these units to the $k$
agents with the highest non-negative ironed virtual values.  Ties,
which occur in ironed virtual surplus maximization when two (or more)
agents' values are mapped to same ironed virtual value, are broken
randomly.  Ex post, we can classify the agents into at most
three groups: those that win with certainty (winners), those that lose
with certainty (losers), and those that win with a common probability
strictly between~0 and~1 (partial winners).


\begin{definition}
A {\em two-level $(\highprice,\lowprice)$-lottery}, denoted
$\Lottery_{\highprice,\lowprice}$, first serves agents with values
strictly more than~$\highprice$, then serves agents with values
strictly more than~$\lowprice$, while supplies last (breaking ties
randomly, as needed).  All agents with
values at most $\lowprice$ are rejected.  
\end{definition}

It will be useful to calculate explicitly, using Lemma~\ref{l:ic}, the
payments of a two-level lottery.  
Let $S$ and $T$ denote the sets of agents with value in the ranges
$(\highprice,\infty)$ and $(\lowprice,p]$, respectively.  Let $s =
  \setsize{S}$ and $t = 
  \setsize{T}$.  For simplicity,  assume that $s \leq k < s+t$, where
  $k$ is the 
  number of units available.  The payments are as follows.
\begin{enumerate}
\item Agents $i \in S$ are each allocated a unit and charged
\begin{align}
\label{eq:highpayment}
\pricei &= \highprice - (\highprice - \lowprice) \, \tfrac{k-s+1}{t+1}.
\end{align}
\item The remaining $k-s$ units are allocated uniformly at random to
  the $k-s$ agents $i \in T$, i.e., by lottery; each such winner pays
  $\pricei = \lowprice$.
\end{enumerate}

We characterize the performance benchmark for platform design for
general linear objectives in terms of two-level lotteries.

\begin{theorem}
\label{t:benchmark}
$\bm(\vals) = \sup\nolimits_\dist \Mye_\dist(\vals) =
\sup_{\highprice,\lowprice} \Lottery_{\highprice,\lowprice}(\vals)$.
\end{theorem}

\begin{proof} 
The outcome of ironed virtual surplus
maximization is equivalent to a $k$-unit
$(\highprice,\lowprice)$-lottery.  To see this, consider an ironed
virtual valuation function $\ironmarg$ and a valuation profile
$\vals$.  Set $\highprice$ to be the infimum bid that the
highest-valued agent can make and be a winner (possibly larger than
the agent's value), and
$\lowprice$ to be the infimum bid that a partial winner can make and
remain a partial winner (or $\highprice$ if there are no partial
winners).  The two mechanisms have the same outcome on profile $\vals$.
Conversely, every $(\highprice,\lowprice)$-lottery arises in ironed
virtual surplus maximization with respect to some i.i.d.\
distribution, for example
with $\ironmarg(\val) = 2$ for $\val \in (\highprice,\infty)$,
$\ironmarg(\val) = 1$ for $\val \in (\lowprice,\highprice]$, and
  $\ironmarg(\val) = -1$ for $\val \leq \lowprice$.\footnote{For
    objectives like residual surplus where the virtual values are
    always non-negative, set $\ironmarg(\val) = 1/2$ instead of $-1$
    for $\val \leq q$.  See the construction in
    Appendix~\ref{app:distribution-construction} for details.}  
\end{proof}

We conclude with a simple but useful observation: The values of
$\highprice$ and $\lowprice$ that attain the supremum in
Theorem~\ref{t:benchmark} must each either be zero, infinity, or an
agent's value.  Observe that the objective $\sum_i \valweight \vali \alloci +
\priceweight \pricei$ is linear in payments.  If $\lowprice$ or
$\highprice$ is not in the valuation profile, then it can either be
increased or decreased without decreasing the objective.  For example,
lowering $\highprice$ or $\lowprice$ without changing the allocation
increases residual surplus.



\section{Residual Surplus}\label{sec:worst}
\label{sec:residual-surplus}

In this section we consider platform design for the objective of
residual surplus.  We consider separately the $n=2$ agent case
and the general $n>2$ agent case.  For $n=2$ agents (and a single
item) we completely 
execute our template for platform design by reinterpreting the
benchmark, giving a platform mechanism that is universally adopted
with competitive advantage~$4/3$, and proving that no platform
mechanism is universally adopted with a smaller competitive
advantage.  The platform mechanism that achieves this bound is 
neither a standard auction nor a mixture over standard auctions, where
by ``standard'' we mean a symmetric Bayesian-optimal mechanism with
respect to some valuation distribution.

For every number $n>2$ of agents and $k \ge 1$ of items, we give a
heuristic platform that guarantees universal adoption with a constant
competitive advantage (independent of~$k$, $n$, and the support of the
valuations).  This platform is not a mixture of standard auctions, and
we show in Appendix~\ref{app:standard-auction-lb} that no such mixture
is universally adopted with any finite competitive advantage (as $n
\rightarrow \infty$).  This heuristic mechanism identifies properties
of good platforms and is a proof-of-concept that good platforms exist.

\subsection{Single-unit Two-agent Platforms}
\label{subsec:n=2}

We now execute the framework for platform design for two agents, a
single unit, and the objective of residual surplus.  Bayesian optimal
mechanisms and our benchmark are characterized in
Sections~\ref{sec:bayesian} and~\ref{sec:benchmark}, respectively; for
two agents and a single item, the benchmark takes a simple form.

There are only two relevant $(\highprice,\lowprice)$-lotteries for the
performance benchmark, the degenerate $\highprice=\lowprice =0$
lottery, and the $\highprice = \valith[2]$ and $\lowprice = 0$
lottery; here $\valith[1]$ and $\valith[2]$ denote the highest and
second-highest agent values, respectively.  From
equation~\eqref{eq:highpayment}, the residual surpluses of these
two-level lotteries are $\tfrac{\vali[1] + \vali[2]}{2}$ (i.e., the
average value) and $\valith[1] - \tfrac{\valith[2]}{2}$, respectively.
Thus,
\begin{align}
\label{eq:benchmark-n=2}
\bm(\vals) &= \max \{ \tfrac{\vali[1] + \vali[2]}{2}, \valith[1] - \tfrac{\valith[2]}{2}\}.
\end{align}
This benchmark is depicted in Figure~\ref{fig:benchmark-n=2}.

We now turn to the problem of designing a platform mechanism that is
universally adopted with a minimal competitive advantage.  As
mentioned above, the lottery is adopted with a competitive advantage
of~2.  A natural approach to platform design is to randomly mix over
two platforms that are good in different settings.  
For example, the Vickrey auction is good on the valuation
profile $\vals = (1,0)$, whereas the lottery is good on the valuation
profile $\vals = (1,1)$.  Considering only these two valuation
profiles (where $\bm(\vals) = 1$), choosing the Vickrey auction with
probability~$1/3$ and the lottery with probability~$2/3$ balances the
competitive advantage necessary for adoption of the platform for each
profile at~$3/2$.  In fact, a routine calculation shows that this
mixture is universally adopted with competitive advantage~$3/2$.  This
platform mechanism is, however, not optimal.

One approach to solving for the optimal platform mechanism is to look
for a mechanism that achieves the same approximation factor to the
benchmark for every valuation profile.\footnote{Our
  optimal platform for monopoly pricing in
  Section~\ref{sec:monop-pricing} also exhibits this property.}
Inspecting the benchmark (Figure~\ref{fig:benchmark-n=2}), we conclude
that an auction with identical approximation factor on all inputs must
have a discontinuity in behavior only where the ratio between the high
and low value is~2.  Importantly, there should be no discontinuity in
behavior when the values are equal, that is, the optimal platform
should never mix over the Vickrey auction.  These observations suggest
the following parameterized class of auctions.

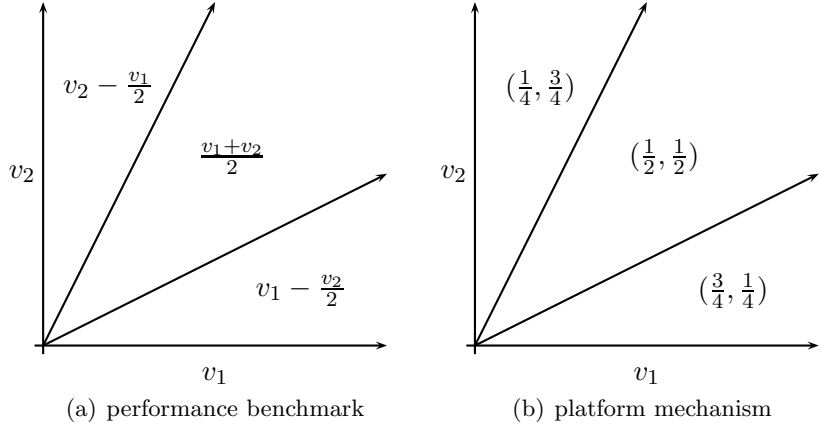
\begin{figure}
\psset{unit=.9in}
\begin{center}
\subfigure[performance benchmark]{\label{fig:benchmark-n=2}%
\begin{pspicture}(-.2,-.2)(2.2,2)

\psline{->}(-.05,0)(2,0)
\psline{->}(0,-.05)(0,2)

\psline{->}(0,0)(2,1)
\psline{->}(0,0)(1,2)

\rput(.38,1.5){$\vali[2]-\tfrac{\vali[1]}{2}$}
\rput(1.1,1.1){$\tfrac{\vali[1]+\vali[2]}{2}$}
\rput(1.5,.32){$\vali[1]-\tfrac{\vali[2]}{2}$}

\rput[l](-.2,1){$\vali[2]$}
\rput[B](1,-.2){$\vali[1]$}

\end{pspicture}}
\subfigure[platform mechanism]{\label{fig:platform-n=2}
\begin{pspicture}(-.2,-.2)(2.2,2)

\psline{->}(-.05,0)(2,0)
\psline{->}(0,-.05)(0,2)

\psline{->}(0,0)(2,1)
\psline{->}(0,0)(1,2)

\rput(.38,1.5){$(\tfrac14,\tfrac34)$}
\rput(1.1,1.1){$(\tfrac12,\tfrac12)$}
\rput(1.5,.32){$(\tfrac34,\tfrac14)$}

\rput[l](-.2,1){$\vali[2]$}
\rput[B](1,-.2){$\vali[1]$}

\end{pspicture}}
\end{center}
\caption{The performance benchmark~\eqref{eq:benchmark-n=2} and
  optimal platform mechanism 
  for the single-item, two-agent, residual-surplus-maximization problem.
  The positive quadrant is partitioned by the lines $\vali[1] =
  2\vali[2]$ and $2\vali[1] = \vali[2]$.  The allocation rule of the
  platform mechanism is given as $(\alloci[1],\alloci[2])$.}
\label{fig:benchmark-platform-n=2}
\end{figure}

\begin{definition} 
The two-agent single-item {\em ratio auction with
ratio $\ratio \geq 1$ and bias $\bias \in [1/2,1]$} allocates the good
according to a fair coin if the agent values are within a factor
$\ratio$ of each other and, otherwise, according to a biased coin with
probability $\bias$ in favor of 
the high-value agent.\footnote{Appropriate payments can be derived by
  reinterpreting the ratio auction as a distribution over weighted
  Vickrey auctions; see also the proof of
  Lemma~\ref{l:optimal-platform-n=2}.} 
\end{definition}

The Vickrey auction and the lottery are special cases of the ratio auction.
With bias $1/2$ the ratio auction is a lottery
(for every ratio); with bias $\bias = 1$ and ratio $\ratio = 1$ it
is the Vickrey auction.  We next show that the optimal two-agent
single-item platform for 
residual surplus is the ratio auction with ratio $\ratio
= 2$ and bias $\bias = 3/4$.  The allocation probabilities of this
auction are depicted in Figure~\ref{fig:platform-n=2}.  It is adopted
with competitive advantage~$4/3$.


\begin{lemma}
\label{l:optimal-platform-n=2}
The ratio auction with ratio $\ratio = 2$ and bias $\bias = 3/4$
is universally adopted with competitive advantage~$4/3$.
\end{lemma}

\begin{proof}
The ratio auction (with ratio $\ratio$) can always be expressed as a
distribution over weighted Vickrey auctions, where $w_1 = 1$, $w_2$
is selected randomly from some distribution over the set
$\{0,1/\ratio,\ratio,\infty\}$, and the agent~$i$ that maximizes
$w_i\vali$ winning the item.  With bias $\bias = 3/4$, the
distribution over the set is uniform.  We calculate the auction's
approximation of the benchmark via simple case analysis.  the expected
residual surplus from the four choices of $w_2$ averages to
$\frac{1}{4}\left[\vali[1] + (\vali[1] - \frac{\vali[2]}{2}) +
  (\vali[2] - \frac{\vali[1]}{2}) + \vali[2]\right] = \frac{3}{4}
\tfrac{\vali[1] + \vali[2]}{2}$ when $\vali[1] \in
      [\vali[2]/2,2\vali[2]]$ and to $\frac{1}{4}\left[\vali[1] +\-
        (\vali[1] - \frac{\vali[2]}{2}) + (\vali[1] - 2\vali[2]) +
        \vali[2]\right] = \frac{3}{4}\left(\vali[1] -
      \frac{\vali[2]}{2}\right)$ when $\vali[1] > 2\vali[2]$.  The
      case where $\vali[1] < \vali[2]/2$ is symmetric.  In each case,
      the expected residual surplus is exactly $\frac{3}{4}
      \bm(\vals)$.
\end{proof}

We now show that the ratio auction with ratio $\ratio = 2$ and bias
$\bias = 3/4$ is an optimal platform; meaning, no platform is
universally adopted with competitive advantage less than~$4/3$.  We
first note that, for every distribution~$\dist$, the expected residual
surplus of the ratio auction with ratio $\ratio = 2$ and bias $\bias =
3/4$ is exactly $3/4$ times the expected value of the benchmark~$\bm$.
Of course, the Bayesian optimal auction for~$\dist$ is no worse.

\begin{corollary}\label{cor:43}
For every distribution $\dist$ and $n=2$ agents and $k=1$ item, the
expected benchmark is at most~$4/3$ times the expected residual
surplus of the optimal auction, that is, $\expect{\bm(\vals)} \leq
\frac{4}{3}\expect{\Mye_\dist(\vals)}.$
\end{corollary}
The following technical lemma exhibits a distribution~$\dist$ for
which the inequality in Corollary~\ref{cor:43} is tight.  Intuitively,
this distribution is the one with constant virtual value for utility.

\begin{lemma}
\label{l:lower-bound-n=2}
For the exponential distribution $\dist(z)= 1-e^{-z}$, $n=2$
agents, $k=1$ unit, the expected value of the benchmark is~$4/3$
times the expected residual surplus of the optimal auction, that is, 
$\expect{\bm(\vals)} = \frac{4}{3}\expect{\Mye_\dist(\vals)}$.
\end{lemma}

\begin{proof}
Since the exponential distribution has a monotone hazard rate, a
 lottery maximizes the expected residual surplus (Section~\ref{sec:bayesian}).  The expected value of an
exponential random variable is~1 so 
$\expect{\Mye_\dist(\vals)} = \expect{\val} = 1$.

We now calculate the expected value of the benchmark $\bm(\vals)$
defined in equation~\eqref{eq:benchmark-n=2}.  Write the smaller value as $\val
= \valith[2]$ and the higher value as $x + \val = \valith[1]$ for $x
\geq 0$.  In terms of $\val$ and $x$ the benchmark is $\val +
\frac{x}{2}$ when $x \leq \val$ and $\frac{\val}{2} + x$ when $x \geq
\val$.  Therefore, the expectation of $\bm$ conditioned on $\val$ is
\begin{align*}
\expect{\bm(x+\val,\val) \given \val}
& = 
\int_0^{\val} \left(\val+\tfrac{x}{2} \right)\, e^{-x}\,dx
+ 
\int_{\val}^{\infty} \left(\tfrac{v}{2} + x\right)\,e^{-x}\,dx\\
& = 
\val(1-e^{-\val}) + \tfrac{1}{2} \left( 1 - (\val + 1)e^{-\val} \right)
+ \tfrac{\val}{2}e^{-\val} + (\val + 1)e^{-\val}\\
& = 
\val + \tfrac{1}{2} \left( 1 + e^{-\val} \right).\\
\intertext{The smaller value~$\valith[2] = \val$ is distributed
  according to an exponential distribution with rate~2.
Integrating out yields}
\expect{\bm(x + \val,\val)}
& = 
\int_{0}^{\infty} \left(\val + \tfrac{1}{2} + \tfrac{1}{2}e^{-\val}
\right)\, 2e^{-2\val}\,d\val\\
& = 
\tfrac{1}{2} + \tfrac{1}{2} + \int_0^{\infty} e^{-3\val}\,d\val
 =  \tfrac{4}{3}. \qquad\qquad\qquad\qedhere
\end{align*}
\end{proof}


For the setting of Lemma~\ref{l:lower-bound-n=2}, the optimal mechanism
has expected residual surplus $\frac{3}{4}\expect{\bm(\vals)}$.  Any
platform mechanism is only worse and, by the definition of
expectation, there must be a valuation profile $\vals$ where this
platform mechanism has residual surplus at most $\frac{3}{4}
\bm(\vals)$.


\begin{corollary}
For $n\geq 2$ agents, $k=1$ item, and the residual surplus objective,
no platform mechanism is universally adopted with competitive
advantage less than~$4/3$.
\end{corollary}

We conclude that the ratio auction with ratio $\ratio = 2$ and bias
$\bias = 3/4$ is an optimal platform for two-agent, single-item residual
surplus maximization.

\subsection{Multi-unit, Multi-agent Platforms}
\label{subsec:rsol}

\newcommand{\sm}{{\setminus}}
\newcommand{\sse}{{\subseteq}}
\newcommand{\event}{{\cal E}}
\newcommand{\lb}{{ ??}}

We now turn to markets with $n>2$ agents and $k \ge 1$ units.
We show that the minimum competitive advantage for universal adoption
is a finite constant, independent of the number of units, the number
of bidders, and the support size of the valuations. 

In contrast to the $n=2$ case, neither the Vickrey auction, the
lottery, nor a convex combination thereof obtains a constant-factor
approximation of the benchmark~$\bm$ (Definition~\ref{d:benchmark}).  For instance,
with one object and valuation profile $\vals = (1,1,0,\ldots,0)$, the
Vickrey auction has zero residual surplus and the lottery has expected
residual surplus $2/n$, while the benchmark residual surplus is
$\bm(\vals) =1$.  In fact, no Bayesian optimal auction (a.k.a.,
standard auction) or mixture over standard auctions is universally
adopted with a competitive advantage that is an absolute constant.
This result is stated as Theorem~\ref{t:standard-auction-lb}, below,
and proved in Appendix~\ref{app:standard-auction-lb}.  We conclude that
the derivation of a platform mechanism that is universally adopted
with a constant competitive advantage requires non-standard auction
techniques.

\begin{theorem} 
\label{t:standard-auction-lb}
For every $\rho > 1$ there is a sufficiently large~$n$ such
that, for an $n$-agent, 1-unit setting, no mixture over standard
auctions is universally adopted with competitive advantage $\rho$.
\end{theorem}

Due to the complexity of the problem, we relax the goal of determining
the optimal platform mechanism and instead look for a heuristic
platform that is universally adopted with a constant competitive
advantage.  We believe this heuristic pinpoints properties of good
platforms, while the optimal platform is complex and
perhaps difficult to interpret.

This heuristic follows the random sampling paradigm of \citet{GHW-01}.
Half of the agents (henceforth: sample) are used for a market analysis
to determine a good mechanism to run on the other half of the
agents (henceforth: market).  We do not attempt to estimate the
distribution of the sample, as distributions are complex objects.
Instead, we use the sample to determine a good two-level lottery and
then simply run that two-level lottery on the market.  Two-level
lotteries are described by two numbers and are therefore,
statistically, far simpler objects than distributions.

To make this task even simpler, we first argue that two-level
lotteries can be approximated by one-level lotteries.  

\begin{definition} 
The {\em one-level $\oneprice$-lottery}, denoted $\Lottery_{\oneprice}$, serves agents
with values strictly more than $\oneprice$, while supplies last
(breaking ties randomly).
Winners are charged $\oneprice$ and agents with values below
$\oneprice$ are rejected.
\end{definition}

\begin{lemma} \label{lem:lotteries}
For every valuation profile $\vals$ and parameters $k$, $\highprice$,
and $\lowprice$, there is an $\oneprice$ such that the $k$-unit
$\oneprice$-lottery obtains at least half of the expected residual
surplus of the $k$-unit $(\highprice,\lowprice)$-lottery.
\end{lemma}

\begin{proof}
We prove the lemma by showing that $\Lottery_{\highprice,\lowprice}(\vals) \le
\Lottery_{\highprice}(\vals) + \Lottery_{\lowprice}(\vals)$.  We argue the stronger
statement that each agent enjoys at least as large a combined expected
utility in $\Lottery_{\highprice}(\vals)$ and $\Lottery_{\lowprice}(\vals)$ as in
$\Lottery_{\highprice,\lowprice}(\vals)$.

Let $S$ and $T$ denote the agents with values in the ranges
$(\highprice,\infty)$ and $(\lowprice,\highprice]$, respectively.  Let
  $s = \setsize{S}$ and $t = \setsize{T}$.  Assume that $0 < s \leq k
  < s+t$ as otherwise the $k$-unit $(\highprice,\lowprice)$ lottery is
equivalent to a one-level lottery.  Each agent in~$T$ participates in a
  $k$-unit $\lowprice$-lottery in $\Lottery_{\lowprice}$ and only a
  $(k-s)$-unit $\lowprice$-lottery in $\Lottery_{p,q}$; her expected
  utility can only be smaller in the second case.  Now consider $i \in
  S$.  Writing $\lowprob = (k-s+1)/(t+1)$ in
  equation~\eqref{eq:highpayment} we can upper bound the utility of
  agent $i$ in $\Lottery_{\highprice,\lowprice}$ by
$$v_i - \highprice + \lowprob (\highprice-\lowprice) = 
(1-\lowprob)(v_i-\highprice) + \lowprob(v_i-\lowprice)
\le (v_i-\highprice) + \tfrac{k}{s+t} \cdot (v_i-\lowprice),$$
which is the combined expected utility that the agent obtains from
participating in both a $k$-unit $\highprice$-lottery (with $s \leq k$) and a
$k$-unit $\lowprice$-lottery.
\end{proof}

\begin{corollary} \label{cor:lottery} 
For every valuation profile $\vals$, the benchmark~$\bm$ is at most
twice the expected residual surplus of the best one-level lottery:
$\bm(\vals) \leq 2 \cdot \sup_\oneprice \Lottery_{\oneprice}(\vals)$.
\end{corollary}

The following auction does market analysis on the fly to identify and
run a good one-level lottery.  We have deliberately avoided optimizing
the parameters of this mechanism in order to keep its description and analysis as simple
as possible.

\begin{definition}
\label{def:rsol}
The $k$-unit {\em Random Sampling Optimal Lottery (RSOL)} mechanism
works as follows.
\begin{enumerate}

\item Partition the agents uniformly at random into a market $M$ and a
  sample $S$, i.e., each agent is in $S$ or $M$ independently with
  probability $1/2$ each. 

\item Calculate the optimal $k$-unit lottery price $\oneprice_S$ for
  the sample: $\oneprice_S = \argmax_\oneprice \Lottery_{\oneprice}
  (\vals_S)$.

\item Run the $k$-unit $\oneprice_S$-lottery on the market $M$; reject
  the agents in the sample $S$.

\end{enumerate}
\end{definition}

We show that this RSOL mechanism gives a good
approximation to the residual surplus of the optimal one-level lottery
unless a majority of its residual surplus is derived from the
highest-valued agent.  If a majority of its residual
surplus is derived from the highest-valued agent, then the $k$-unit
Vickrey auction is a good approximation of the benchmark.  Therefore,
mixing between 
the two auctions gives a platform that is universally adopted with
constant competitive advantage (independent of $k$ and $n$).

\begin{theorem}\label{thm:worst}
For every $n,k \ge 1$,
there is an $n$-agent $k$-unit platform mechanism that is universally
adopted with constant competitive advantage.
\end{theorem}

A key fact that enables the analysis of RSOL is that, with
constant probability, the relevant statistical properties of the
full valuation profile are preserved in the market and the sample.  These
statistical properties can be summarized in terms of a
``balance'' condition.  Define a partition of the
agents~$\{1,2,3,\ldots,n\}$ into a market $M$ and 
a sample $S$ to be {\em balanced} if $1 \in M$, $2 \in S$, and for all $i
\in \{3,\ldots,n\}$, between $i/4$ and $3i/4$
of the $i$ highest-valued
agents are in $S$ (and similarly~$M$).  In the proof of
Theorem~\ref{thm:worst}, we use the following adaptation of the
``Balanced Sampling Lemma'' of \citet{FFHK-05} to bound from below 
the probability that $\RSOL$ selects a balanced partitioning.  

\begin{lemma}
\label{l:balanced-sampling}
When each agent is assigned to the market~$M$ or sample~$S$ independently
according to a fair coin, the resulting partitioning is balanced with
probability at least~$0.169$.
\end{lemma}

For completeness, we include a
proof of Lemma~\ref{l:balanced-sampling} in
Appendix~\ref{app:balanced-sampling}.  We now turn to
Theorem~\ref{thm:worst}.

\begin{proof}[Proof of Theorem~\ref{thm:worst}]
We outline the high-level argument and then fill in the details.
We focus on the expected residual surplus of RSOL, where the
expectation is over the random partition of agents, relative to
that of an optimal one-level lottery, on the ``truncated''
valuation profile $\vals \super 2 =
(\valith[2],\valith[2],\valith[3],\ldots,\valith[n])$.  
We only track the contributions to RSOL's expected residual surplus when the
partitioning of the agents is balanced.  In such cases, RSOL's residual
surplus on the truncated valuation profile can only be less than on
the original one.

Step 1 of the analysis proves that, conditioned on the partitioning
of the agents being balanced, the expected residual surplus of the
optimal one-level lottery for the sample is at least~$1/2$ times that
of the optimal one-level lottery for the full truncated valuation
profile.  Step 2 of the analysis proves that, conditioned on an
arbitrary balanced partition, the residual surplus of every one-level
lottery on the market is at least~$1/9$ times its residual surplus on
the sample.  In particular, this inequality holds for the optimal
one-level lottery for the sample.  Combining these two steps with
Lemma~\ref{l:balanced-sampling} implies that the expected residual
surplus of RSOL is at least $0.169 \times \tfrac{1}{2} \times
\tfrac{1}{9} \ge 1/107$ times that of the optimal one-level lottery on
the truncated valuation profile~$\vals \super 2$.  The additional
residual surplus achieved by an optimal one-level lottery on the
original valuation profile~$\vals$ over the truncated one is at most
$\valith[1] - \valith[2]$.  The residual surplus of the $(k+1)$th-price
auction, where~$k$ is the number of units for sale, is at least this
amount.  The platform mechanism that mixes between RSOL with
probability $107/108$ and the $(k+1)$th-price auction with probability
$1/108$ has expected residual surplus at least~$1/108$ times that of
the optimal one-level lottery on $\vals$, and (by
Corollary~\ref{cor:lottery}) at least~$1/216$ times the benchmark
$\bm$.  Below, we elaborate on the two steps described above.

Step 1: {\em Conditioned on a balanced partitioning, the expected
  residual surplus of the optimal one-level lottery for the sample~$S$
  is at least~$1/2$ times that of the optimal one-level lottery for
  the full truncated valuation profile.}
Let $\oneprice$ be the price of the optimal one-level lottery for
$\vals \super 2$.  Conditioned on a balanced partition, exactly one of the top two
(equal-valued) bidders of $\vals \super 2$ lies in~$S$.  By symmetry,
each other bidder has probability~$1/2$ of lying in~$S$.
The winning probability of bidders in~$S$ with value at
least~$\oneprice$ is only higher than that when all agents are present.
Summing over the bidders' contributions to the residual surplus and
using the linearity of expectation, 
$\expect[S]{\Lottery_\oneprice(S) \given \text{balanced partition} } \geq
\Lottery_\oneprice(\vals \super 2)/2$.  Of course, the optimal one-level
lottery for the sample is only better.

Step 2: {\em Conditioned on an arbitrary balanced partition, for the
  truncated valuation profile~$\vals \super 2$, the
  residual surplus of every one-level lottery on the market is at
  least~$1/9$ times its residual surplus on the sample.}
Fix a balanced partition into~$S$ and~$M$ and a one-level lottery at
price $\oneprice$.  
%
The expected contribution of a bidder~$j$ to a $\oneprice$-lottery is
$(\val_j-\oneprice)$ times its winning probability (if $\val_j >
\oneprice$) or~0 (otherwise).  The balance condition ensures that, for
every~$i \ge 2$, the number of the $i$ highest-valued bidders that belong to
the market is between~$1/3$ and~$3$ times that of the sample.  In
particular, the winning probability of bidders with value at least
$\oneprice$ in~$M$ is at least~$1/3$ of that of such bidders in~$S$.
Moreover, the balance condition implies that
$$
\sum\nolimits_{j \in M} \max \{ \val_j-\oneprice,0 \}
\ge
\tfrac{1}{3} \sum\nolimits_{j \in S} \max \{ \val_j-\oneprice,0 \}
$$
for the truncated valuation profile~$\vals \super 2$;
the claim follows.
\end{proof}

It is certainly possible to optimize better the parameters of the
platform mechanism defined in the proof of Theorem~\ref{thm:worst}.
Furthermore, since for simplicity we only keep track of RSOL's
performance when the partition is balanced, the mechanism's
performance is better than the proved bound.

\section{Platform Design and Prior-Free Profit Maximization}
\label{sec:profit}

While the objective of profit maximization is not central to this paper,
there have been a number of studies of prior-free mechanisms
for profit maximization that are relevant to platform design.  
This section discusses digital good settings
(Section~\ref{subsec:profit-digital-good}), multi-unit settings
(Section~\ref{subsec:profit-multi-unit}), and more general settings
(Section~\ref{subsec:profit-general}).  We describe these results
using the terminology of platform design.  An important goal of
our discussion is to compare our performance benchmark, which is
justified by Bayesian foundations, with the prior-free benchmarks
employed in this literature.

\subsection{Digital Good Settings}
\label{subsec:profit-digital-good}

The simplest setting for platform design is that of a {\em digital
  good}, i.e., a multi-unit setting with the same number $k=n$ of
units as (unit-demand) agents.  This environment admits a trivial
optimal mechanism for surplus and residual surplus (serve all agents
for free); but for profit maximization, designing a good platform
mechanism is a challenging problem.

%
%
The Bayesian optimal mechanism for a digital good when values are
drawn i.i.d.~from the distribution~$\dist$ simply posts the monopoly
price for $\dist$, i.e., an $\oneprice$ that maximizes $\oneprice
(1-\dist(\oneprice))$.  In the language of the preceding sections, this
optimal mechanism can be viewed as an $\oneprice$-lottery.  The
performance benchmark described in Section~\ref{sec:benchmark}
simplifies to
\begin{align}
\label{eq:benchmark-digital-goods}
\bm(\vals) &= \max\nolimits_i i \valith.
\end{align}

For $n=1$ agent, the benchmark \eqref{eq:benchmark-digital-goods}
equals the surplus and, as we concluded in
Section~\ref{sec:monop-pricing}, it cannot be well approximated by any
platform mechanism.  Because of this technicality, the benchmark $\bm
\super 2$ to which prior-free digital good auctions have been compared
\citep[e.g.,][]{GHKSW-06} explicitly excludes the possibility of
deriving all its profit from one agent:
\begin{align}
\label{eq:benchmark2-digital-goods}
\bm \super 2(\vals) &= \max\nolimits_{i\geq 2} i \valith.
\end{align}

Therefore, up to the technical difference between
benchmarks~\eqref{eq:benchmark-digital-goods}
and~\eqref{eq:benchmark2-digital-goods}, the prior-free literature for
digital goods is compatible with our framework for platform design.
Some notable results in this literature are as follows.  For reasons
we explain shortly, we refer to the
approximation of $\bms$ as giving near-universal adoption.  Optimal
platform mechanisms are given in \citet{GHKSW-06} and \citet{HM-05}
for two and three-player digital goods settings, where the competitive
advantages for near-universal adoption are precisely 2 and $13/6$,
respectively.  
As the number $n$ of agents tends to infinity, \citet{GHKSW-06} show
that there is no platform mechanism that is near-universally adopted with
competitive advantage less than 2.42; and \citet{CGL-14} show that
there exists a mechanism that matches this bound.  This optimal
platform mechanism is fairly complex; \citet{HM-05}
had previously given a simple mechanism that is near-universally adopted
with competitive advantage 3.25.

The benchmark $\bm \super 2$ does not satisfy our most basic
requirement for benchmarks: there exist distributions for which the
expected Bayesian optimal revenue exceeds the expected value of the
benchmark.\footnote{For example, consider $n$ agents, each having
  value~$n^2$ with probability $1/n^2$ and 0 otherwise.  The expected
  revenue of a Bayesian optimal mechanism is~$n$.  The expected value
  of the benchmark $\bm \super 2$ is bounded above by a constant,
independent of~$n$.}
Therefore, mechanisms that approximate the benchmark may not be
universally adopted.  This problem is not an artifact of the $\bm
\super 2$ benchmark and is inherent to profit
maximization: pathological distributions show that
there is {\em no} benchmark $\bm'$ and constant $\beta$ that satisfy $\beta
\expect[\vals]{\Mye_\dist(\vals)} \geq \expect[\vals]{\bm'(\vals)}
\geq \expect[\vals]{\Mye_\dist(\vals)}$ for every distribution $\dist$.

For the profit objective, the requirement of universal adoption can be
relaxed to adoption for every distribution in a
large permissive class of distributions.  Approximation
of the benchmark $\bms$ implies such a near-universal adoption, in the
following sense. 
\begin{proposition}
\label{prop:near-universal-adoption}
If mechanism $\mech$ is a $\beta$ approximation to $\bms$ on all
valuation profiles, i.e., $\mech(\vals) \geq \bms(\vals)/\beta$ for
every $\vals$, then $\mech$ is adopted with competitive advantage
$\beta$ on distributions $\dist$ with
$\expect[\vals]{\bms(\vals)} \geq \expect[\vals]{\Mye_\dist(\vals)}$.
\end{proposition}

Proposition~\ref{prop:near-universal-adoption} has bite in that it is
satisfied by most relevant distributions.  The following lemma, which
we prove in Appendix~\ref{app:profit-benchmark-lb}, gives a sufficient
condition for the distribution.  Intuitively, this condition states
that the revenue from posting a price does not drop too quickly
as that price is lowered, and it is a strict generalization of
the regularity condition of \citet{mye-81}.  
This condition is not satisfied in the bad example above,
as most of the optimal revenue is derived from one high-valued agent.

\begin{lemma}
\label{l:bms>mye}
For digital good settings and every distribution $\dist$ with
$\val\,(1-\dist(\val))/\dist(\val)$ non-increasing,
$\expect[\vals]{\bms(\vals)} \geq \expect[\vals]{\Mye_\dist(\vals)}$.
\end{lemma}

\subsection{Multi-unit Settings}
\label{subsec:profit-multi-unit}

We next consider maximizing profit in a $k$-unit auction with
unit-demand bidders.  We assume throughout that $k \ge 2$.
We next define a variant of the performance benchmark $\bms$ for
platform design and compare it to the benchmark $\ofs(\vals) = \max_{2 \leq
  i \leq k} i \valith$ that has been employed, without formal
justification, in previous work on prior-free multi-unit auctions.

The benchmark $\bm$ defined as the supremum of Bayesian optimal
mechanisms is, by Theorem~\ref{t:benchmark}, equivalent to the
supremum over two-level lotteries (which need not sell all units).
Two-level lotteries are not useful for profit-maximization in 
digital goods settings, where a $(\highprice,\lowprice)$-lottery is
equivalent to a $\lowprice$-lottery which is equivalent to a
$\lowprice$ price posting.  They are useful in limited
supply settings, however.
The performance benchmark~$\bms$ is defined by $\bms(\vals)
= \bm(\valith[2],\valith[2],\valith[3],\ldots,\valith[n])$.
For every valuation profile,
the benchmark $\bm \super 2$ is at most twice the value of $\ofs$
(cf.~Lemma~\ref{lem:lotteries}).  Thus,
every multi-unit auction that $\beta$-approximates the
benchmark~$\ofs$ also $2\beta$-approximates the benchmark~$\bms$.  
As in Section~\ref{subsec:profit-digital-good}, approximation of the
benchmark~$\bms$ implies approximation of the optimal expected revenue
in every Bayesian setting with a non-pathological distribution.

The above discussion provides Bayesian foundations for the benchmark
$\ofs$, and translates the known results for approximating that
benchmark into good platform designs.  Specifically, every digital
good auction that $\beta$-approximates the $\ofs$ benchmark ---
equivalently for a digital good, the $\bms$ benchmark --- can be
easily converted into a multi-unit auction that $\beta$-approximates
the $\ofs$ benchmark~\citep{GHKSW-06} and hence $2\beta$-approximates
the $\bms$ benchmark.  The optimal platform mechanism for digital
goods from \citet{CGL-14} can be thus converted into a platform
mechanism for limited supply that is near-universally adopted with
competitive advantage 4.84, i.e., it is a 4.84-approximation to
$\bms$.


%
%
%
%

\subsection{General Environments}
\label{subsec:profit-general}

The approach to Bayesian optimal mechanism design discussed in
Section~\ref{sec:bayesian} characterizes optimal mechanisms beyond
just multi-unit settings.  For every
single-parameter setting, where agents want service and there is a
feasibility constraint over the set of agents that can be
simultaneously served, the optimal mechanism is the ironed virtual
surplus maximizer.

Our benchmark~$\bm$ is difficult to analyze beyond multi-unit settings.  In
follow-up work to this paper, \citet{HY-11} gave a refinement
of our benchmark using the notion of envy-freedom.  For instance, when
the set system that constrains feasible outcomes satisfies a 
substitutes condition (formally: the set of feasible outcomes are the
independent sets of a matroid set system), their benchmark is at least
as large as ours.  They give a mechanism that is similar to the RSOL
(Definition~\ref{def:rsol}) that, for these set systems, is
near-universally adopted with constant competitive advantage.

\section{Discussion}\label{sec:conc}

We defined an analysis framework for platform design based on relative
approximation of a performance benchmark.  Auctions that approximate
this benchmark are simultaneously near-optimal in every Bayesian
setting with i.i.d.\ bidder valuations.  Optimizing within this
analysis framework suggests novel multi-unit auction formats,
different from those suggested by Bayesian analysis.  The framework is
flexible and permits several extensions and modifications, discussed
next.

We focused on platform design for the objectives of residual
surplus (Section~\ref{sec:worst}) and profit maximization
(Section~\ref{sec:profit}), but our platform design approach extends
beyond these objectives.  As an example, imagine the $k$-unit auction in
an i.i.d.~Bayesian setting where the optimal solution is characterized
by optimizing the ironed virtual value corresponding to ``a 8\%
government sales tax.''  
The objective is then the value of the agents and
mechanism less the tax deducted by government, and 
the corresponding virtual value function has the form
$\virt(\val) = 0.92 \val - 0.08 \tfrac{1-\dist(\val)}{\dens(\val)}$.
The optimal $k$-unit $(\highprice,\lowprice)$-priority lottery
remains the appropriate benchmark for this and every other linear objective.
For every linear objective with $\priceweight < 0$, see
equation~\eqref{eq:objective},
there is always an optimal $(\highprice,\lowprice)$-priority lottery
with $\highprice$ and $\lowprice$ at most the second highest value
$\valith[2]$, and the mechanism RSOL of Section~\ref{subsec:rsol}
can be mixed with a Vickrey auction to approximate the benchmark.


Optimal mechanisms are ironed virtual surplus maximizers in every
single-parameter Bayesian setting with independent private
values~\citep{mye-81}, not just in the multi-unit auction settings
studied here.  Examples of more general single-parameter settings
include constrained matching markets, single-minded combinatorial
auctions, and public projects.  The performance benchmark
(Definition~\ref{d:benchmark}) can again be defined pointwise as the
supremum over the performance of ironed virtual surplus maximizers on
a given valuation profile.  As discussed in
Section~\ref{subsec:profit-general}, this performance benchmark seems
hard to characterize beyond multi-unit settings (cf.,
Theorem~\ref{t:benchmark}).  The follow-up work of \citet{HY-11}
recently proposed an alternative benchmark based on envy freedom.
This benchmark has structure similar to that of Bayesian optimal
auctions and, for this reason, is analytically tractable.
\citet{HY-11}, for the profit objective and the envy-free benchmark,
give platform mechanisms that are near-universally adopted with
constant competitive advantage in general settings.


We defined the performance benchmark as the supremum performance of
(symmetric) mechanisms that are optimal in some Bayesian setting with
i.i.d.\ valuations.  Clearly, one can define such a benchmark with
respect to {\em any} class of mechanisms.  As an example alternative,
consider residual surplus maximization and the class of symmetric
mechanisms that are Bayesian optimal for some i.i.d..\ distribution
that is either MHR or anti-MHR --- that is, the class consisting
solely of the Vickrey auction and the (zero-price) lottery.  This
class arises naturally when domain knowledge suggests that only
MHR and anti-MHR distributions are relevant, or if outside consultants
are only equipped to design optimal mechanisms for these cases.
Specializing to the two-bidder single-item case studied in
Section~\ref{subsec:n=2}, the platform design benchmark decreases from
$\max \{ \tfrac{\vali[1] + \vali[2]}{2}, \valith[1] -
\tfrac{\valith[2]}{2}\}$ to $\max \{ \tfrac{\vali[1] + \vali[2]}{2},
\valith[1] - \valith[2]\}$.  Reworking the analysis of that section
for this new benchmark shows that the optimal mechanism remains a ratio
auction, just with a different setting of the parameters (namely,
$\ratio = 3$ and $\bias = 4/5$, for an approximation ratio of~$5/4$).
Notably, the format of the optimal platform is robust to this
particular change in the benchmark.

The mechanism in Section~\ref{subsec:rsol} demonstrates the
existence of platform mechanisms for residual surplus maximization
that are universally adopted with constant competitive advantage,
independent of the number of units and bidders.  No standard auction,
meaning an ironed virtual surplus maximizer,
enjoys such a guarantee (Appendix~\ref{app:standard-auction-lb}).
Developing our understanding of $n$-player platform design further is
an interesting research direction.  For starters, there should be a
much tighter analysis of our mixture of RSOL and Vickrey auctions.
One avenue for improvement is to track
contributions to the residual surplus on the 83\% of the probability
space for which the chosen partition is not balanced.  The ``average
balance'' approach of \citet{AMS-09}, used previously to improve over
the balanced partition technique of \citet{FFHK-05} in a
profit-maximization context, can be used to give such an improved
bound.  A second idea is to compare
the mechanism's residual surplus directly to the benchmark in
Theorem~\ref{t:benchmark}, rather than to the simpler ``approximate
benchmark'' in Corollary~\ref{cor:lottery}.
Analogously, avoiding the
approximate benchmark~$\ofs$ could lead to better profit-maximizing
platform designs (see Section~\ref{subsec:profit-multi-unit}).

There are surely platforms that are universally adopted with smaller
competitive advantage than is required by 
that in Section~\ref{subsec:rsol}.  A challenging problem
is to characterize optimal platforms for the residual surplus
objective.  Even the case of three-bidder
single-item settings appears challenging.
We conjecture that, for residual surplus maximization
with $n$ bidders and a single item, the minimum competitive advantage
required by an optimal platform for universal adoption is precisely
the expected value of the performance benchmark 
when bidders' valuations are drawn
i.i.d.\ from the exponential distribution
(as in Theorem~\ref{t:benchmark}).

While solving for optimal platforms is interesting theoretically, we
suspect that optimal platforms will suffer from some drawbacks.
First, when the number of agents is large, the optimal platform is a
complex object, perhaps a distribution over a very large number of
different auctions.  This complexity is characteristic of exact
optimization in any auction analysis framework; other well-known
examples include profit-maximizing auctions in Bayesian single-item
settings when bidders' valuations are not identical or are
i.i.d.\ from an irregular distribution~\citep{mye-81}.  The complexity
of optimal auctions motivates the design and analysis of platforms
that are relatively simple while requiring a competitive advantage
that is almost as small as the minimum possible, e.g., in the spirit
of~\citet{BK-96} and~\cite{HR-09}.  Second, in optimizing a min-max
criterion, the optimal platform will equalize the approximation factor
of the benchmark across all valuation profiles (cf., the proof of
Lemma~\ref{l:optimal-platform-n=2}).  In practice there might be
agreed-upon ``common inputs'' and ``rare inputs'' --- without there
necessarily being a full prior --- with auction performance on common
inputs being the most important.  For example, the
random-sampling-based auctions of \citet{BBHM-08} out perform the
optimal platform mechanism \citep{GHKSW-06} on a family of
common inputs.

An auction that approximates the performance benchmark is
simultaneously near-optimal in every Bayesian setting with
i.i.d.\ bidder valuations.  The converse need not hold, and an
interesting research direction is to better understand the
relationship between these two conditions.  Sometimes, as with the
monopoly pricing problem studied in Section~\ref{sec:monop-pricing},
simultaneous Bayesian near-optimality is as hard as approximation of
the performance benchmark.\footnote{For a value of~$h \ge 1$, consider
  the set of distributions that are concentrated at a single point
  in~$[1,h]$.  For each such distribution, the corresponding optimal
  auction extracts full surplus.  As in
  Section~\ref{sec:monop-pricing}, no single auction can extract more
  than a~$1/\ln h$ fraction of the surplus for every such
  distribution.}  This is not always the case, however.  For example,
\citet{DRY-10} showed that the digital good auction that partitions
the agents into pairs and runs a Vickrey auction to serve one agent in
each pair obtains a 2-approximation to the revenue of the Bayesian
optimal mechanism whenever the distribution is {\em regular}, meaning
that virtual values are increasing (cf.~Section~\ref{sec:bayesian} and
Appendix~\ref{app:profit-benchmark-lb}).  The proof of this
2-approximation is a simple consequence of the $n=1$ special case of
the main theorem of \citet{BK-96}, i.e., that the 2-agent Vickrey
auction obtains more revenue than monopoly pricing a single agent.  As
mentioned in Section~\ref{subsec:profit-digital-good}, no auction for
a digital good achieves a 2-approximation of the
benchmark~$\bms$~\citep{GHKSW-06}.  On the other hand, all work thus
far on simultaneous Bayesian near-optimality that avoids the pointwise
benchmark approach --- termed ``prior-independent guarantees'' by
\citet{DRY-10} --- are confined to regular
distributions~\citep{DRY-10,DHKN-11,RTY-12}.  By contrast, our
benchmark approximations directly imply prior-independent guarantees
for most distributions (for profit maximization) and for all
distributions (for residual surplus maximization).

\appendix

\section{Distribution Construction from Virtual Values}
\label{app:distribution-construction}

Section~\ref{sec:bayesian} gives a formula for calculating an
agent's virtual value function from the distribution from which her
value is drawn.  For the residual surplus objective, this formula is
$\marg(\val) =  
\frac{1-\dist(\val)}{\dens(\val)}$.  
This section reverses the calculation and gives a
constructive proof that every non-negative piecewise constant
function arises as the virtual value (for utility) function for some
distribution.
This fact is alluded to in Section~\ref{sec:benchmark} and is used
explicitly in the proof of Theorem~\ref{t:standard-auction-lb} in
Appendix~\ref{app:standard-auction-lb}.

First observe that the exponential distribution has constant virtual
value equal to its mean.  That is, the exponential distribution with
mean $\exval$ has rate $1/\exval$, cumulative distribution
$\dist_\exval (z) = 1 - e^{-z/\exval}$, and virtual value
$\marg_{\exval}(z) = \exval$.  

Now consider a non-negative piecewise constant function $\pwcfunc :
[0,\infty) \to \reals_+$, where the boundaries of each interval are
  given by $\interval_0=0,\interval_1,\ldots$, and where the value of
  the function on the interval $[\interval_j,\interval_{j+1})$ is
    $\pwcfunc_j$.  We construct the distribution $\dist$ with
    virtual value function $\pwcfunc(\cdot)$ inductively.  The
    starting interval is given by 
    the distribution function $\dist(z) = \dist_{\pwcfunc_0}(z)$ for
 $z \in [\interval_0,\interval_1]$.
    With the first $j-1$ intervals defined by $\dist(z)$ for $z \in
    [\interval_{0},\interval_j]$, we define the distribution function
    for the $j$th interval 
by $\dist(\interval_j+z) =
    \dist_{\pwcfunc_j}(\dist_{\pwcfunc_j}^{-1}(\dist(\interval_j)) +
    z)$
for $\interval_j+z \in
    [\interval_j,\interval_{j+1}]$.
Intuitively, this construction does a horizontal shift of the
    distribution function of the exponential distribution with mean
    $\exval_j$ so 
    that its height matches the height of the constructed (so far)
    distribution function at $\interval_j$.  This construction is
    illustrated in Figure~\ref{fig:dist-construction}.

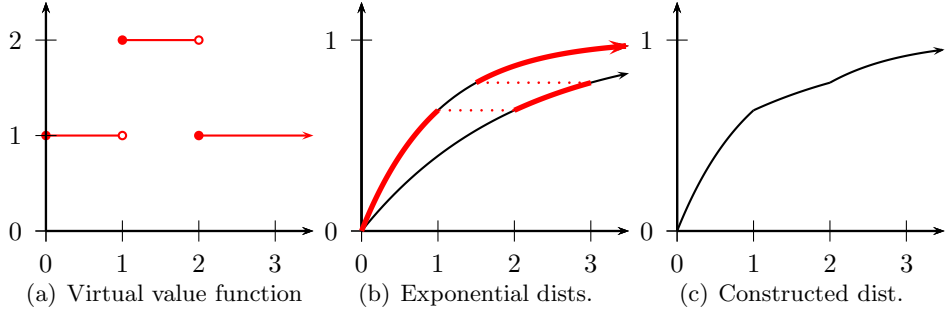
\begin{figure}[t]
\small
\psset{xunit=.4in,yunit=.5in}
\begin{center}
\subfigure[Virtual value function]{
\begin{pspicture}(-.4,-.4)(3.5,2.4)


\psline[linecolor=red]{*-o}(0,1)(1,1)
\psline[linecolor=red]{*-o}(1,2)(2,2)
\psline[linecolor=red]{*->}(2,1)(3.5,1)

\psaxes{->}(3.5,2.4)

\end{pspicture}}
\psset{yunit=1in}
\subfigure[Exponential dists.]
{
\begin{pspicture}(-.4,-.2)(3.5,1.2)

\psaxes{->}(3.5,1.2)

\psplot[arrows=->]{0}{3.5}{1 2.71828182846 x neg exp sub}
\psplot[arrows=->]{0}{3.5}{1 2.71828182846 x .5 mul neg exp sub}

\psset{linewidth=2pt,linecolor=red}
\psplot{0}{1}{1 2.71828182846 x neg exp sub}
\psplot[arrows=->]{1.5}{3.5}{1 2.71828182846 x neg exp sub}

\psplot{2}{3}{1 2.71828182846 x .5 mul neg exp sub}

\psset{linecolor=red,linewidth=1pt,linestyle=dotted}
\psline(1,.6321)(2,.6321)
\psline(1.5,.7769)(3,.7769)

\end{pspicture}}
\subfigure[Constructed dist.]
{
\begin{pspicture}(-.4,-.2)(3.5,1.2)

\psplot{0}{1}{1 2.71828182846 x neg exp sub}
\psplot{1}{2}{1 2.71828182846 x 1 add .5 mul neg exp sub}
\psplot[arrows=->]{2}{3.5}{1 2.71828182846 x .5 sub neg exp sub}

\psaxes{->}(3.5,1.2)

\end{pspicture}}
\end{center}
\caption{Construction of the cumulative distribution function $\dist$ 
  with virtual value function (for utility) that is
  piecewise constant on $[0,1)$, $[1,2)$, and $[2,\infty)$, with
        virtual values 1, 2, and 1, respectively.}  
\label{fig:dist-construction}
\end{figure}

\section{Inadequacy of Standard Auctions}
\label{app:standard-auction-lb}

This section establishes the limitations of {\em standard auctions}, i.e.,
ironed virtual value maximizers (including lotteries and the Vickrey
auction), as platform mechanisms.  Section~\ref{subsec:n=2} shows that
with $n=2$ agents, the optimal platform mechanism is not a mixture of
standard auctions.  Here we show that, even with $k=1$ unit, there is
no finite competitive advantage for which a mixture of standard
auctions is universally adopted for all $n$.  This result contrasts
with Theorem~\ref{thm:worst}, which gives a (non-standard) platform
mechanism that is universally adopted with a constant competitive
advantage, independent of $n$ and $k$.


Our argument uses the distributions in the construction in
Appendix~\ref{app:distribution-construction}; we next note some of
their salient properties.  These distributions are piece-wise
exponential distributions, with piece-wise constant virtual values for
utility and piece-wise constant hazard rates.  Recall that the virtual
value for an exponential distribution equals its expected value which
equals the reciprocal of its hazard rate.  Also, exponential
distributions are memoryless: given that the value $\val$ from an
exponential distribution is at least $z$, the conditional distribution
of $\val$ is identical to that of $z+w$ where $w$ is exponential with
the same rate.  Of particular relevance, the probability that an
exponential random variable with mean one exceeds $\lbparam$ is
$e^{-\lbparam}$, and the probability that an exponential random
variable with mean $\lbparam$ exceeds $\lbparam$ is $1/e$.  For
piece-wise exponential distributions, these properties hold within
each piece.  From the analysis of Section~\ref{sec:bayesian}, the
expected residual surplus of any mechanism $\mech$ on distribution
$\dist$ is equal to its expected virtual surplus.

\begin{numberedtheorem}{\ref{t:standard-auction-lb}}
For every $\rho > 1$ there is a sufficiently large~$n$ such
that, for an $n$-agent, 1-unit setting, no mixture over standard
auctions is universally adopted with competitive advantage $\rho$.
\end{numberedtheorem}

\begin{proof}
Define $\lbparam$ to be an integer greater than or equal to $\max \{
24 \cdot \rho, 2\}$.  For $\lbind \in \{0,1,2,\ldots,\lbparam-1\}$,
let $\dist_{\lbind,\lbparam}$ denote the piece-wise exponential
distribution (as in Appendix~\ref{app:distribution-construction}) with
virtual value for utility equal to~1 everywhere except on the interval
$[\lbind \lbparam,\lbind\lbparam+\lbparam)$, where it is equal to
  $\lbparam$.  Such a distribution thus has a ``special interval''
  where the hazard rate is relatively low (and virual value is
  relatively high).  Let $\C_{\lbparam}$ denote the set of
  distributions $\{ \dist_{\lbind,\lbparam}
  \}_{\lbind=0}^{\lbparam-1}$.  Consider a setting with a single item
  and $n=e^{\beta^2}$ agents (rounded up to the nearest integer).

We claim the following:
\begin{enumerate}[(a)]

\item For every $\dist_{\lbind,\lbparam} \in \C_{\lbparam}$,
  there is an auction with expected residual surplus at least
  $\lbparam/4$.

\item For every standard auction~$\A$, if
  $\dist_{\lbind,\lbparam} \in \C_{\lbparam}$ is chosen uniformly at
  random, then the expected residual surplus of $\A$, over the choice
  of $\dist_{\lbind,\lbparam}$ and valuations $\val_1,\ldots,\val_n
  \sim \dist_{\lbind,\lbparam}$, is at most~6.

\end{enumerate}

Claims~(a) and~(b) imply the theorem.  To see this, the residual
surplus of a convex combination of mechanisms $\mech$ is the convex
combination of their residual surpluses.  As the inequality of
property~(b) holds for each auction in the support of such a convex
combination, it also holds for the combination.
Taking expectation over $\dist$ uniform from $\C_{\lbparam}$ in
inequality~\eqref{eq:bm2} from Section~\ref{sec:benchmark} we have,
\begin{align*}
\expect[\vals,\dist]{\bm(\vals)} 
  &\geq \expect[\vals,\dist]{\Mye_{\dist}(\vals)}
  \geq \tfrac{\beta}{24}\expect[\vals,\dist]{\mech(\vals)}.
\end{align*}
By the definition of expectation, there must exist a valuation profile
$\vals$ that achieves this separation, i.e., with $\bm(\vals) \geq
\frac{\beta}{24}\mech(\vals)$.  Thus, the competitive advantage needed
for universal adoption is at least $\beta/24 \geq \rho$.

We now proceed to the proofs of~(a) and~(b).
Call an agent {\em high-valued}
if her value is at least $\lbparam^2$.
The probability that there is no high-valued agent is at most
$1/\lbparam$.\footnote{The analysis is elementary.  Using the
  memoryless property of exponential distributions, the
  probability that a given agent is high-valued is
$\eta =
(e^{-\lbparam})^{\lbparam-1} \cdot e^{-1} = e^{-\lbparam^2 + (\lbparam
  -1)}$.  With $n = e^{\lbparam^2} = e^{\lbparam-1}/\eta$ agents, the
probability of no high-valued agents is $(1-\eta)^n \leq e^{-\eta n} \le
  e^{-e^{\lbparam-1}} \leq 1/\lbparam$.  The last inequality can be
verified by checking the lower endpoint of $\lbparam = 2$ and
comparing the   derivatives of $e^{e^{\lbparam-1}}$ and $\lbparam$.} 

To prove~(a), fix a choice of $\lbind \in \{0,1,\ldots,\lbparam-1\}$
and consider the $\lbind\lbparam$-lottery.  The probability that all
agents' values are below $\lbind\lbparam$ is less than the probability
that all agents' values are below $\lbparam^2$; thus the probability
of a winner in this lottery is at least $1-1/\lbparam \ge 1/2$.
Otherwise, the winner is a random agent with value at least
$\lbind\lbparam$.  By the memoryless property of exponential
distributions, the probability that the winner, with value at least
$\lbind\lbparam$, has value less than $\lbind\lbparam + \lbparam$ is
$1-1/e \geq 1/2$; such an winner has virtual value $\lbparam$.
Virtual values for residual surplus are non-negative, so the expected
virtual surplus of the $\lbind\lbparam$-lottery is at least
$\tfrac{1}{2} \cdot \tfrac{1}{2} \cdot \lbparam$, as claimed.

To prove~(b), we first warm up by considering the case where $\A$ is
an $\oneprice$-lottery.  Choose $\dist_{\lbind,\lbparam} \in
\C_{\lbparam}$ uniformly at random.  The intuition is that this random
choice effectively ``hides'' the location of the large virtual values.

If $\oneprice \ge \lbind \lbparam + \lbparam$, then the winner of $\A$
(if any) has virtual value~1.
If $\oneprice \le \lbind \lbparam - \lbparam$, then by the memoryless
property of exponential distributions, the value of a winner is less
than $\lbind \lbparam$ with probability at least $1-e^{-\lbparam}$.
Thus, the expected virtual value of a winner in this case (if any) is
at most $1 + \lbparam e^{-\lbparam} \leq 2$.
Finally, if $\oneprice \in (\lbind \lbparam - \lbparam, \lbind
\lbparam + \lbparam)$,
then the virtual value of a winner (if any) is at most $\lbparam$.
As the third case occurs with probability at most $2/\lbparam$ (over
the random choice of $\lbind$), the expected virtual value of $\A$ is
at most $\tfrac{2}{\beta} \cdot \beta + 1 \cdot 2 = 4$.


We conclude by extending the argument of the preceding paragraph for
one-level lotteries to an arbitrary ironed virtual surplus
maximizer~$\A$.  In our single-item symmetric setting, the
auction~$\A$ specifies ironed intervals where ties are broken
randomly but otherwise awards the item to the agent with the highest
value.  Let $[\oneprice,\oneprice']$ denote the ironed interval
of~$\A$ that contains the value $\lbparam^2$.  (If
value~$\lbparam^2$ is not ironed, then set $\oneprice = \oneprice' =
\lbparam^2$.)  Valuation profiles without a high-valued bidder occur
with probability at most $1/\lbparam$ (by the above analysis) and give
virtual surplus at most $\beta$; thus their contribution to the
expected virtual surplus of $\A$ is at most 1.  Valuation profiles
with a high-valued bidder and highest value in
$[\oneprice,\oneprice']$ contribute the same expected virtual surplus
to $\A$ as to a $\oneprice$-lottery; by the previous paragraph,
this contribution is at most~4.  In every valuation profile with
highest value greater than $\oneprice'$, the item is awarded to a
bidder with virtual value~1; these profiles contribute at most~1 to
the expected virtual surplus of~$\A$.  
%
\end{proof}

This proof, in fact, gives a lower bound on the competitive advantage
for universal adoption of any standard auction that grows
proporionally to $\sqrt{\log n}$ with the number $n$ of agents.

\section{The Balanced Sampling Lemma}

\label{app:balanced-sampling}

Recall that a partitioning of the agents~$\{1,2,3,\ldots,n\}$ into a
market $M$ and sample $S$ is {\em balanced} if $1 \in M$, $2 \in S$,
and for all $i \ge 3$, between $i/4$ and $3i/4$ of the $i$th
highest-valued agents are in~$M$ (and similarly~$S$).  We restate and
prove the balanced sampling lemma below.

\begin{numberedlemma}{\ref{l:balanced-sampling}}
When each agent is assigned to the market~$M$ or sample~$S$ independently
according to fair coin, the resulting partitioning is balanced with
probability at least~$0.169$.  
\end{numberedlemma}

\begin{proof}
Call a subset of the agents {\em imbalanced} if, for some~$i \ge 3$,
it contains fewer than~$i/4$ of the $i$ highest-valued agents. After
conditioning on the events that $1 \in M$ and $2 \in S$, the
probability that~$S$ is imbalanced can be calculated as at most
$0.161$ by a simple probability of ruin analysis proposed in
\citet{FFHK-05} (details given below).  By symmetry, the same bound
holds for~$M$.  By the union bound, the partition is balanced with
probability at least $0.678$.  Agent 1 is in $M$ and 2 is in $S$ with
probability~$1/4$, so the unconditional probability that the partition
is balanced is at least $0.169$.

The following analysis from \citet{FFHK-05} shows that the conditional
probability 
that $S$ is imbalanced is at most $0.161$.  Consider the random
variable $Z_i = 4\setsize{S \cap \{1,\ldots,i\}} - i$; the balanced
condition is equivalent to $Z_i \geq 0$ for all $i \geq 3$.  By the
conditioning, $S \cap \{1,2\} = \{2\}$ and so $Z_2 = 2$.
View $Z_i$ as the positions of a random walk on the integers that
starts from position two and takes three steps forward (at step $i$ with $i
\in S$) or one step back (at step $i$ with $i\not\in S$), each with
probability~$1/2$.  The set~$S$ is imbalanced if and only if this
random walk visits position~$-1$.
The probability $r$ of ever (with $n \rightarrow \infty$)
visiting the preceding position in such a random walk can be calculated as the
root of $r^4-2r+1$ on the interval $(0,1)$, which is approximately
$0.544$.  
The probability of imbalance, which requires eventually moving backward
three steps from position~2, is at most $r^3 \leq 0.161$, as claimed.
\end{proof}

\section{Profit Maximization with Near-universal Adoption}
\label{app:profit-benchmark-lb}

Recall from Section~\ref{sec:profit} the benchmark $\bms = \max_{i\geq
  2} i \valith$, which effectively excludes selling to the
highest-valued agent at her value.
We will show that a mechanism $\mech$ that achieves a
$\beta$-approximation of this benchmark on every valuation profile is
near-universally adopted with competitive advantage $\beta$, meaning
that for every distribution $\dist$ in a large class, the expected
profit of~$\mech$ is at least a $\beta$ fraction of that of the
Bayesian optimal auction for $\dist$.  By
Proposition~\ref{prop:near-universal-adoption}, it suffices to give
sufficient condition on $\dist$ that guarantees that
$\expect[\vals]{\bms(\vals)} \geq \expect[\vals]{\Mye_\dist(\vals)}$.

From \citet{BR-89}, virtual values for revenue are given by the {\em
  marginal revenue} of the {\em revenue curve} that plots the revenue
$\price\, (1-\dist(\price))$ against the probability $1-\dist(\price)$
that the agent buys (i.e., her expected demand).  Virtual values are
given by the slope of the revenue curve, thus monotonicity of virtual
values, as required by the regularity condition of \citet{mye-81}, is
equivalent to the concavity of the revenue curve.

Our sufficient condition, the ``inscribed triangle property,'' states
that for every point $(1-\dist(\price), \price\,(1-\dist(\price)))$ on
the revenue curve, the triangle formed with the points $(0,0)$ and
$(1,0)$ lies underneath the revenue curve.
This condition is clearly
satisfied whenever the revenue curve is concave --- equivalently,
whenever the distribution is regular --- and 
is also satisfied by a large family of multi-modal distributions that
are not regular.

To understand this condition better, observe that, for every
distribution $\dist$ and price $\price$, the line from~$(0,0)$ to
$(1-\dist(\price), \price\,(1-\dist(\price)))$ lies beneath the
revenue curve.  The reason is that, for every $\alpha \in [0,1]$, the
price $p'$ with selling probability $\alpha \cdot (1-\dist(\price))$
is at least $p$ and hence obtains revenue $p'(1-\dist(\price')) \ge
\alpha \cdot p(1-\dist(\price))$.  The inscribed triangle property is
therefore equivalent to requiring that, for every price $\price$, the
line between $(1-\dist(\price), \price\,(1-\dist(\price)))$ and
$(1,0)$ lies beneath the revenue curve.  For an economic
interpretation of this condition, consider the measure of types that
are not served at a price $\price$, i.e., $\dist(\price)$.  Viewing
the revenue curve as a function of $\dist(\price)$, the condition says
that as the price is dropped, the revenue per unit of types that are
not served is non-decreasing.  In other words, the condition requires
$\price\,(1-\dist(\price))/\dist(\price)$ to be non-increasing.

The inscribed triangle property immediately implies the following
lemma, which is reminiscent of the main theorem of \citet{BK-96}.

\begin{lemma}
\label{l:inscribed-triangle}
For distribution $\dist$ with non-increasing
$\val\,(1-\dist(\val))/\dist(\val)$, the two-agent
Vickrey auction revenue exceeds the single-agent optimal revenue.
\end{lemma}

\begin{proof}
In the two-agent Vickrey auction, each agent faces a
take-it-or-leave-it offer equal to the other agent's bid.
Thus, each agent faces a random price
$\price$ distributed such that the probability of sale to  this agent is
uniform on $[0,1]$.  
Since the revenue of every such price is given by the revenue curve,
and the distribution of $1-\dist(\price)$ is uniform, the expected
revenue obtained from the agent
equals the area under the revenue curve.
Invoking the inscribed triangle property at the point 
$(1-\dist(\price^*), \price^*\,(1-\dist(\price^*)))$ 
for the monopoly price $\price^*$,
we conclude that the expected revenue obtained from one agent is at
least $\tfrac{1}{2} \cdot 1 \cdot \price^*\,(1-\dist(\price^*))$,
half the expected revenue of the monopoly price.
Since there are two agents, the total expected Vickrey revenue is at
least the optimal single-agent revenue.
\end{proof}

\begin{lemma}
\label{l:conditioning-lemma}
For distribution $\dist$ with the non-increasing
$\val\,(1-\dist(\val))/\dist(\val)$ property, conditioning to exceed a
price $\price$ preserves the property.
\end{lemma}

\begin{proof}
The original condition is saying that the virtual value is not
more negative than the slope of the line that connects that point on
the revenue curve to $(1,0)$.  
After conditioning on being at least $\price$,
the condition can
be viewed on the original revenue curve as the virtual value not being
more negative than the slope of line that connects the point on the
revenue curve to $(1-\dist(\price),0)$.  As this slope is steeper than
the slope of the line through $(1,0)$, the property is preserved
by such conditioning.
\end{proof}

We now combine the above lemmas to prove Lemma~\ref{l:bms>mye},
restated below.  A key intuition in this proof is that
Lemma~\ref{l:inscribed-triangle} implies that
$\expect[\vals]{2\valith[2]} \geq \expect[\vals]{\Mye_\dist(\vals)}$
for $k=2$ items and $n=2$ agents --- the left-hand side 
is double the revenue of the Vickrey auction (with 1 item) and the
right-hand side is double the revenue of the single-agent optimal
mechanism.

\begin{numberedlemma}{\ref{l:bms>mye}}
For digital good settings and every distribution $\dist$ with
$\val\,(1-\dist(\val))/\dist(\val)$ non-increasing,
$\expect[\vals]{\bms(\vals)} \geq \expect[\vals]{\Mye_\dist(\vals)}$.
\end{numberedlemma}

\begin{proof}
Let $\price^* = \argmax_{\price} \price \, (1-\dist(\price))$ be the monopoly
price for the distribution.  The analysis proceeds by conditioning
on $\valith[3] = z$ and considering the cases where $z \leq
\price^*$ and $z > \price^*$.  In the first case,
 we have
\begin{align*}
\expect[\vals]{\smash{\bms(\vals) \given \valith[3] = z \leq \price^*}} 
  & \geq \expect[\vals]{2 \valith[2] \given \valith[3] = z \leq \price^*}\\ 
  & \geq \expect[\vals]{\Mye_\dist(\vals) \given \valith[3] = z \leq \price^*}.
\intertext{The last inequality follows by
  Lemmas~\ref{l:inscribed-triangle} and~\ref{l:conditioning-lemma} and
  the fact that, given $\valith[3] = z \leq \price^*$, $\Mye_\dist$ is
  an auction that sells to at most agents~1 and~2 and therefore
  has revenue that is at most the optimal auction that sells to these
  agents for the conditional distribution.  In the second case, let
  $k^* \geq 3$ be a random variable for the number of units sold by
  $\Mye_\dist$; we have}
\expect[\vals]{\smash{\bms(\vals) \given \valith[3] = z \geq \price^*}} 
  & \geq \expect[\vals]{k^* \valith[k^*] \given \valith[3] = z \geq \price^*}\\ 
  & \geq \expect[\vals]{k^* \price^* \given \valith[3] = z \geq \price^*}\\ 
  & =  \expect[\vals]{\Mye_{\dist}(\vals) \given \valith[3] = z \geq \price^*}.
\end{align*}
Combining the two cases proves the lemma.
\end{proof}

\end{document}